\documentclass[twocolumn]{autart}    % Enable this line and disable the 
\usepackage{graphicx}
\usepackage{amsmath,amsfonts,amssymb,color}
\usepackage{tikz}
\usepackage{pmat}
\usepackage{empheq}
\usepackage{epsfig} % for postscript graphics files
\usepackage{algorithm}
\usepackage{algorithmic}
\usepackage{epsfig}
\usepackage{array}
\usepackage{multirow}
\usepackage{epstopdf}
\usepackage{tikz}
\usetikzlibrary{shapes,snakes,arrows}
\usepackage{cite}
\usepackage[hidelinks]{hyperref}
\makeatletter
\DeclareRobustCommand{\qed}{%
  \ifmmode % if math mode, assume display: omit penalty etc.
  \else \leavevmode\unskip\penalty9999 \hbox{}\nobreak\hfill
  \fi
  \quad\hbox{\qedsymbol}}
\newcommand{\openbox}{\leavevmode
  \hbox to.77778em{%
  \hfil\vrule
  \vbox to.675em{\hrule width.6em\vfil\hrule}%
  \vrule\hfil}}
\newcommand{\qedsymbol}{\openbox}
\newenvironment{proof}[1][\proofname]{\par
  \normalfont
  \topsep6\p@\@plus6\p@ \trivlist
  \item[\hskip\labelsep\itshape
    #1.]\ignorespaces
}{%
  \qed\endtrivlist
}
\newcommand{\proofname}{Proof}
\makeatother
\begin{document}
\newtheorem{theorem}{Theorem}
\newtheorem{problem}{Problem}
\newtheorem{condition}{Condition}
\newtheorem{definition}{Definition}
\newtheorem{lemma}{Lemma}
\newtheorem{proposition}{Proposition}
\newtheorem{corollary}{Corollary}
\newtheorem{remark}{Remark}
\newtheorem{assumption}{Assumption}
\newtheorem{example}{Example}

\begin{frontmatter}

\title{Byzantine-Resilient Distributed Observers for LTI Systems \thanksref{footnoteinfo}} % Title, preferably not more 

\thanks[footnoteinfo]{This paper was not presented at any IFAC 
meeting. This research was supported in part by NSF CAREER award 1653648. Corresponding author: Aritra Mitra. Telephone of the corresponding author: 1-(765)-496-0406.}

\author[auth]{Aritra~Mitra and Shreyas Sundaram} \ead{mitra14@purdue.edu} \ead{sundara2@purdue.edu}  
\address[auth]{School of Electrical and Computer Engineering at Purdue University, West Lafayette, IN 47907, USA}
         
                             % keyword list or with the 
                                          % help of the Automatica 
                                          % keyword wizard

\begin{abstract}                         
Consider a linear time-invariant (LTI) dynamical system monitored by a network of sensors, modeled as nodes of an underlying directed communication graph. We study the problem of collaboratively estimating the state of the system when certain nodes are compromised by adversaries. Specifically, we consider a Byzantine adversary model, where a
compromised node possesses complete knowledge of the system dynamics and the network, and can deviate arbitrarily
from the rules of any prescribed algorithm. We first characterize certain fundamental limitations of any distributed state estimation algorithm in terms of the measurement and communication structure of the nodes. We then develop an attack-resilient, provably correct state estimation algorithm that admits a fully distributed implementation. To characterize feasible network topologies that guarantee success of our proposed technique, we introduce a notion of `strong-robustness' that captures both measurement and communication redundancy. Finally, by drawing connections to bootstrap percolation theory, we argue that given an LTI system and an associated sensor network, the `strong-robustness' property can be checked in polynomial time.
\end{abstract}

\end{frontmatter}

%%%%%%%%%%%%%%%%%%%%%%%%%%%%%%%%%%%%%%%%%%%%%%%%%%%%%%%%%%%%%%%%%%%%%%%%%%%%%%%%
 \section{Introduction}
The control of large-scale complex networked systems such as power grids, transportation networks, and multi-agent robotic systems requires precise estimation of the state of the underlying dynamical process. Typically, in these applications, sensors (nodes) collecting information about the process are scattered over a geographical region. As the diameters of such networks increase, routing  information from all the sensors to a central computational resource induces large delays and creates communication bottlenecks. To bypass these difficulties, it thus becomes important to consider distributed algorithms where individual sensors communicate only with sensors within a given distance. However, the potential merits \cite{survey1} of such a distributed approach are matched by various challenges. In particular, a key challenge is to design networks and distributed algorithms that guarantee reliable operation of the system in the face of faults or sophisticated adversarial attacks on certain sensors. This leads to the motivation behind our present work. 

In a classical distributed state estimation setup, each node receives partial measurements of the state of an LTI process, and seeks to asymptotically estimate the entire state by exchanging information with its neighbors in the network. Given that our primary focus will be on security related issues associated with this problem, we direct the interested reader to recent work on single-time-scale distributed observers in {\cite{martins3,allerton,mitraarchive
,wang,han,kim,Millan,ren}}. For literature on distributed Kalman filtering, see {\cite{infinite1,olfati1,Baras,batti1
,batti2,kamal}}. However, none of these papers address the challenges associated with tolerating unreliable components in the network. Accordingly, we now provide a survey of the cyber-security literature that is most relevant to our present cause.

\textbf{Related Work:} Over the last decade, a significant amount of research has focused on security in networked control systems. In particular, for noiseless dynamical systems, it has been established that zero-dynamics play a key role in characterizing the stealth of an attack \cite{fabio,sundaramtac}. For networked control systems affected by noise, the authors in \cite{bai1} recently introduced an information-theoretic metric that quantifies the detectability of an attack. A unifying feature of \cite{fabio,bai1}, and the ones in {\cite{teixeira,pajic,mishra,mo}}, is that they involve systems where all the sensor measurements are available at a single location. In the sequel, we shall refer to such systems as centralized control systems. Our problem formulation and subsequent analysis differs from the above literature by constraining each sensor to exchange information with only its neighbors in the communication graph. {Some recent related work on resilient distributed parameter estimation and resilient decentralized hypothesis testing are reported in \cite{yuan} and \cite{varshney}, respectively. The authors in \cite{forti} consider the problem of joint attack detection and state estimation. However, the attack model, the system model, and the assumptions on the communication graph in \cite{forti} differ considerably from the ones considered in this paper.}

While the study of security in centralized control systems is now mature, there lacks a comprehensive theoretical understanding of analogous questions in a distributed setting. Preliminary attempts to counter adversarial behavior in a distributed state estimation context are reported in \cite{sec1},\cite{sec3}. However, unlike our results, these papers neither provide any theoretical guarantees of success, nor allude to graph-theoretic conditions that are necessary for their respective algorithms to work. Recently, in \cite{deghat}, the authors employ an $H_{\infty}$ based approach for detecting biasing attacks in distributed estimation networks. Our present work deviates from \cite{deghat} in several aspects, namely (i) while the analysis in \cite{deghat} is limited to a certain class of attack inputs, our attack model allows compromised nodes to behave \textit{arbitrarily}, i.e., no restrictions are placed on the inputs that can be injected by an adversary, (ii) unlike \cite{deghat}, we develop a filtering algorithm that allows each uncompromised node to asymptotically recover the state of the plant \textit{without} explicitly detecting the nodes under attack, and (iii) the existence of the attack detection filter proposed in \cite{deghat} relies on solving an LMI; however, the authors neither provide graph-theoretic insights regarding the solvability of such an LMI nor discuss whether the LMI can be solved in a distributed manner.  In contrast, we detail graph-theoretic conditions that allow each step of our approach to have a resilient, distributed implementation. {Summing up}, this paper attempts to bridge the gap between centralized and distributed resilient state estimation. Our main contributions are discussed below.

\textbf{Contributions:} Our contributions are threefold. First, in Section \ref{sec:fundamental}, we characterize certain necessary conditions that need to be  satisfied by the sensor measurements and the communication graph for the distributed state estimation problem to be solvable in the presence of arbitrary adversarial behavior. Our results hold for \textit{any} algorithm and hence identify fundamental limitations that are of both theoretical and practical importance in the design of attack-resilient robust networks. We also argue that our impossibility results in the distributed setting generalize those existing for centralized control systems subject to sensor attacks \cite{fawzi,joao2}. 

For the problem under consideration, it is imperative to understand which (potentially adversarial) neighbors a given node should listen to, and subsequently, how it should process the information received from neighbors it chooses to listen to. Consequently, our second contribution is to develop a distributed filtering algorithm in Section \ref{section:estimation} that enables each uncompromised node to recover the entire state dynamics, provided certain graph conditions are met. A thorough analysis of the proposed filtering scheme is then presented in Section \ref{sec:analysis}. 

As our third contribution, in Section \ref{sec:feasible}, we introduce a topological property called `strong-robustness' to characterize feasible systems and networks that guarantee applicability of our approach. By drawing connections to bootstrap percolation theory, we show that the `strong-robustness' property can be checked in polynomial time (in the size of the system and the network). 

{\textbf{Comparison with prior work by the authors:} We reported certain preliminary results in \cite{mitraCDC}. In this paper, we significantly expand upon our prior work in the following ways. (i) While the analysis in \cite{mitraCDC} was limited to system matrices with real, distinct eigenvalues, our present framework allows the system matrix to have arbitrary spectrum. This generalization (accounting for complex, possibly repeated eigenvalues) requires various appropriate modifications to the algorithm developed in \cite{mitraCDC}, along with a more detailed technical analysis. (ii) Section \ref{sec:fundamental} is a new addition entirely, and discusses fundamental limitations of any distributed state estimation algorithm in the face of arbitrary adversarial attacks. (iii) Section \ref{sec:feasible} contains additional details about properties of feasible network topologies, and establishes the key result that the topological condition (namely `strong-robustness') needed for implementing our proposed algorithm can be checked in polynomial time. The latter result (missing in \cite{mitraCDC}) is particularly important since it highlights the applicability of our overall approach.}

 \textbf{Notation:} A directed graph is denoted by $\mathcal{G} =(\mathcal{V},\mathcal{E})$, where $\mathcal{V} =\{1, \cdots, N\}$ is the set of nodes and $\mathcal{E} \subseteq \mathcal{V} \times \mathcal{V} $ represents the edges. An edge from node $j$ to node $i$, denoted by (${j,i}$), implies that node $j$ can transmit information to node $i$. The neighborhood of the $i$-th node is defined as $\mathcal{N}_i \triangleq \{j\,|\,(j,i) \in \mathcal{E} \}.$ A node $j$ is said to be an {out-neighbor} of node $i$ if $(i,j)\in\mathcal{E}$. By an \textit{induced} subgraph of $\mathcal{G}$ obtained by removing certain nodes $\mathcal{C} \subset \mathcal{V}$, we refer to the subgraph that has $\mathcal{V}\setminus\mathcal{C}$ as its node set and contains only those edges of $\mathcal{E}$ with both end points in $\mathcal{V}\setminus\mathcal{C}$. The notation $|\mathcal{V}|$ is used to denote the cardinality of a set $\mathcal{V}$. The set of all eigenvalues (or modes) of a matrix $\mathbf{A}$ is denoted by $sp(\mathbf{A}) = \{\lambda \in \mathbb{C}\,|\,det(\mathbf{A}-\lambda\mathbf{I}) = 0\}$, and the set of all unstable eigenvalues by $\Lambda_{U}(\mathbf{A}) = \{\lambda \in sp(\mathbf{A})\,|\, |\lambda| \geq 1 \}$. We use $a_\mathbf{A}(\lambda)$ and $g_\mathbf{A}(\lambda)$ to denote the algebraic and geometric multiplicities, respectively, of an eigenvalue $\lambda \in sp(\mathbf{A})$. An eigenvalue $\lambda$ is said to be simple if $a_\mathbf{A}(\lambda)=g_\mathbf{A}(\lambda)=1$. Given a set of matrices $\{\mathbf{A}_1, \cdots, \mathbf{A}_n\}$, we use $diag(\mathbf{A}_1, \cdots, \mathbf{A}_n)$ to refer to a block diagonal matrix with $\mathbf{A}_i$ as its $i$-th block entry. For a set $\mathcal{J}=\{m_1, \cdots, m_{|\mathcal{J}|}\} \subseteq \{1, \cdots, N\}$, and a matrix $\mathbf{C}={\begin{bmatrix}\mathbf{C}^T_{1} \hspace{1.5mm} & \cdots & \hspace{1.5mm} \mathbf{C}^T_{N}\end{bmatrix}}^{T}$, we define $\mathbf{C}_{\mathcal{J}} \triangleq {\begin{bmatrix}\mathbf{C}^T_{m_1} \hspace{1mm} & \cdots & \hspace{1mm}\mathbf{C}^T_{m_{|\mathcal{J}|}}\end{bmatrix}}^{T}$. The identity matrix of dimension $r$ is denoted $\mathbf{I}_r$, and $\mathbb{N}_{+}$ is used to refer to the set of all positive integers. The terms `communication graph' and `network' are used interchangeably, and the  term `resilient' is used in the same context as that used traditionally in the computer science literature to deal with worst-case adversarial attack models \cite{Byz}.
 
\section{System and Attack Model}
\textbf{System Model:} Consider the LTI dynamical system
\begin{equation}
\mathbf{x}[k+1] = \mathbf{Ax}[k],
\label{eqn:plant}
\end{equation}
where $k \in \mathbb{N}$ is the discrete-time index, $\mathbf{x}[k] \in {\mathbb{R}}^n$ is the state vector and  $\mathbf{A} \in {\mathbb{R}}^{ n \times n} $ is the system matrix. The system is monitored by a network $\mathcal{G}=(\mathcal{V,E})$ consisting of $N$ nodes. The $i$-th node receives a measurement of the state, given by
\begin{equation}
\mathbf{y}_{i}[k]=\mathbf{C}_i\mathbf{x}[k],
\label{eqn:Obsmodel}
\end{equation}
where $\mathbf{y}_{i}[k] \in {\mathbb{R}}^{r_i}$ and $\mathbf{C}_i \in {\mathbb{R}}^{r_i \times n}$. We use $\mathbf{C}={\begin{bmatrix}\mathbf{C}^T_{1} \hspace{1.5mm} & \cdots & \hspace{1.5mm} \mathbf{C}^T_{N}\end{bmatrix}}^{T}$ to represent the collection of the individual node observation matrices; accordingly, $\mathbf{y}[k]={\begin{bmatrix}\mathbf{y}^T_{1}[k] \hspace{1.5mm} & \cdots & \hspace{1.5mm} \mathbf{y}^T_{N}[k]\end{bmatrix}}^T$ represents the {collective} measurement vector, i.e., $\mathbf{y}[k]=\mathbf{C}\mathbf{x}[k]$.

Each node is tasked with estimating the entire system state $\mathbf{x}[k]$ based on information received from its neighbors and its local measurements (if any). As such, we assume that the pair $(\mathbf{A},\mathbf{C})$ is detectable (this is a necessary condition for solving the distributed state estimation problem even in the absence of adversaries); however, we do not assume that the pair $(\mathbf{A,C}_i)$ is detectable for any $i \in \mathcal{V}$. Two immediate challenges are as follows: (i) As the pair $(\mathbf{A},\mathbf{C}_i)$ may not be detectable for some (or all) $i\in \{1, \cdots, N\}$, information exchange is necessary; and (ii) information exchange is restricted by the underlying communication graph $\mathcal{G}$. 
In addition to the above challenges, in this paper, we allow for the possibility that certain nodes in the network are compromised by an adversary, and \textit{do not} follow their prescribed state estimate update rule. We will use the following adversary model in this paper.

\textbf{Adversary Model:} We consider a subset $\mathcal{A} \subset \mathcal{V}$ of the nodes in the network to be adversarial. We assume that the adversarial nodes are completely aware of the network topology, the system dynamics and the algorithm employed by the non-adversarial nodes. Such an assumption of omniscient adversarial behavior is standard in the literature on resilient distributed algorithms \cite{vaidyacons,rescons,
Sundaramopt,su,flocal1,flocal2}, and allows us to provide guarantees against ``worst-case" adversarial behavior. In terms of capabilities, an adversarial node can leverage the aforementioned information to arbitrarily deviate from the rules of any prescribed algorithm, while colluding with other adversaries in the process. Furthermore, following the Byzantine fault model\cite{Byz}, adversaries are allowed to send differing state estimates to different neighbors at the same instant of time. To characterize the threat model in terms of the number of adversaries in the network, we will use the following definitions from \cite{flocal1},\cite{flocal2}.

\begin{definition} (\textbf{$f$-total set})
A set $\mathcal{C} \subset \mathcal{V}$ is \textit{$f$-total} if it contains at most $f$ nodes in the network, i.e., $| \mathcal{C}| \leq f$.
\end{definition}
\begin{definition} (\textbf{$f$-local set})
A set $\mathcal{C} \subset \mathcal{V}$ is \textit{$f$-local} if it contains at most $f$ nodes in the neighborhood of the other nodes, i.e., $|\mathcal{N}_i \cap \mathcal{C}| \leq f,$ $\forall i \in \mathcal{V}\setminus \mathcal{C}$.
\end{definition}
\begin{definition} (\textbf{$f$-local and $f$-total adversarial models}) A set $\mathcal{A}$ of adversarial nodes is \textit{$f$-locally bounded} (resp., $f$-totally bounded) if $\mathcal{A}$ is an $f$-local (resp., $f$-total) set.
\end{definition}
In the literature dealing with distributed fault-tolerant algorithms, it is a common assumption to consider an {$f$-total} adversarial model. However, to allow for a large number of adversaries in large scale networks, we will allow the adversarial set to be $f$-local. Summarily, the adversary model considered throughout this paper will be referred to as an $f$-locally bounded Byzantine adversary model. The non-adversarial nodes will be referred to as regular nodes and be represented by the set $\mathcal{R}=\mathcal{V}\setminus\mathcal{A}$. Note that the actual number and identities of the adversarial nodes are not known to the regular nodes.  As is standard, any reliable system is designed to provide a desired level of resilience against a maximum number of component failures or attacks. We share the same philosophy. Specifically, we assume that each node in the network is programmed to tolerate upto a maximum of $f$ adversaries in the entire network (in an $f$-total model) or in its own neighborhood (in an $f$-local model). Such an assumption is typical in the design of distributed protocols (for varied applications, such as consensus \cite{rescons,vaidyacons}, optimization \cite{su,Sundaramopt}, and broadcasting \cite{flocal1,flocal2}) that are resilient to worst-case Byzantine attack models like the one considered in this paper.

Throughout this paper, we shall only consider causal (i.e., nodes act only on past and present information), synchronous (i.e., all nodes share a common clock w.r.t. their iterates), and deterministic algorithms (i.e., given the same input, such algorithms generate the same output); note however that the notions of causal and deterministic behavior apply only to the regular nodes.  We shall also assume that all quantities being updated iteratively by the regular nodes are initialized identically in each execution. With $\hat{\mathbf{x}}_i[k]$ representing the estimate of $\mathbf{x}[k]$ maintained by node $i$, the problem studied in this paper can be formally stated as follows.
\newpage
\begin{problem} 
\textbf{(Resilient Distributed State Estimation)} Given an LTI system (\ref{eqn:plant}), a linear measurement model (\ref{eqn:Obsmodel}), and a time-invariant directed communication graph $\mathcal{G}$, design a set of state estimate update and information exchange rules such that $\lim_{k\to\infty} ||\hat{\mathbf{x}}_i[k]-\mathbf{x}[k]||=0$, $\forall i \in \mathcal{R}$, \textit{regardless} of the actions of any $f$-locally bounded set of Byzantine adversaries.
\end{problem}

The interplay between the measurement structure of the nodes and the underlying communication graph results in certain conditions being necessary for solving Problem 1, irrespective of the choice of algorithms. We provide such conditions in the following section.

\section{Fundamental Limitations of any Distributed State Estimation Algorithm}
\label{sec:fundamental}
Intuitively, the network must possess a certain degree of measurement redundancy as well as redundancy in its communication structure so as to counteract the effects of adversarial behavior. More specifically, the measurements of the regular nodes must ensure collective detectability of the state, and the network structure should prevent the malicious nodes from acting as bottlenecks between correctly functioning nodes. To identify necessary conditions for resilient distributed state estimation that capture the above notions of redundancy, we first introduce some terminology. 

\begin{figure}[t]
\begin{center}
\begin{tikzpicture}
[->,shorten >=1pt,scale=.75,inner sep=1pt, minimum size=12pt, auto=center, node distance=3cm,
  thick, node/.style={circle, draw=black, thick},]
\node [circle, draw, fill=white](n5) at (0,0.75)  (5)  {5};
\node [circle, draw, fill=white](n4) at (0,2.25)   (4)  {4};
\node [circle, draw, fill=white](n6) at (0,-0.75)   (6)  {6};
\node [circle, draw, fill=white](n7) at (0,-2.25)  (7)  {7};
\node [circle, draw, fill=white](n2) at (-2.5,0)   (2)  {2};
\node [circle, draw, fill=white] (n1) at (-2.5,1.5)   (1)  {1};
\node [circle, draw, fill=white](n3) at (-2.5,-1.5)   (3)  {3};
\node [circle, draw, fill=white](n8) at (2.5,1.5)   (8)  {8};
\node [circle, draw, fill=white](n9) at (2.5,0)   (9)  {9};
\node [circle, draw, fill=white](n10) at (2.5,-1.5)   (10)  {10};
\node [circle, draw, fill=white](ns1) at (-5,0)   (s1)  {$s_1$};
\node [circle, draw, fill=white](ns2) at (5,0)   (s2)  {$s_2$};

\draw[->,blue!50,dashed,very thick]
  (s1) -- (1); 
\draw[->,blue!50,dashed,very thick]
  (s1) -- (2);
\draw[->,blue!50,dashed,very thick]
  (s1) -- (3); 
\draw[->,blue!50,dashed,very thick]
  (s2) -- (8); 
\draw[->,blue!50,dashed,very thick]
  (s2) -- (9);
\draw[->,blue!50,dashed,very thick]
  (s2) -- (10); 
\draw[->,blue!50,dashed,very thick]
  (s1) -- (1); 
\draw[->,blue!50,dashed,very thick]
  (s1) -- (2);
\draw[->,black,very thick]
  (1) -- (2); 
\draw[->,black,very thick]
  (1) -- (4);
\draw[->,black,very thick]
  (1) -- (6);
\draw[->,black,very thick]
  (2) -- (1); 
\draw[->,black,very thick]
  (2) -- (3);
\draw[->,black,very thick]
  (2) -- (4); 
\draw[->,black,very thick]
  (2) -- (5);
\draw[->,black,very thick]
  (2) -- (6);
\draw[->,black,very thick]
  (2) -- (7); 
\draw[->,black,very thick]
  (3) -- (2);
\draw[->,black,very thick]
  (3) -- (5);
\draw[->,black,very thick]
  (3) -- (7);
\draw[->,black,very thick]
  (4) -- (1); 
\draw[->,black,very thick]
  (4) -- (8);
\draw[->,black,very thick]
  (4) -- (5);
\draw[->,black,very thick]
  (5) -- (2);
\draw[->,black,very thick]
  (5) -- (9);
\draw[->,black,very thick]
  (5) -- (4);
\draw[->,black,very thick]
  (5) -- (6);
\draw[->,black,very thick]
  (5) -- (3);
\draw[->,black,very thick]
  (5) -- (10);
\draw[->,black,very thick]
  (6) -- (2); 
\draw[->,black,very thick]
  (6) -- (9);
\draw[->,black,very thick]
  (6) -- (1);
\draw[->,black,very thick]
  (6) -- (8);
\draw[->,black,very thick]
  (6) -- (5);
\draw[->,black,very thick]
  (6) -- (7);
\draw[->,black,very thick]
  (7) -- (3);
\draw[->,black,very thick]
  (7) -- (10);
\draw[->,black,very thick]
  (7) -- (6);
\draw[->,black,very thick]
  (8) -- (4); 
\draw[->,black,very thick]
  (8) -- (6);
\draw[->,black,very thick]
  (8) -- (9);
\draw[->,black,very thick]
  (9) -- (4); 
\draw[->,black,very thick]
  (9) -- (5);
\draw[->,black,very thick]
  (9) -- (6);
\draw[->,black,very thick]
  (9) -- (7); 
\draw[->,black,very thick]
  (9) -- (8);
\draw[->,black,very thick]
  (9) -- (10);
\draw[->,black,very thick]
  (10) -- (5);
\draw[->,black,very thick]
  (10) -- (7);
\draw[->,black,very thick]
  (10) -- (9);
\node (rect) at (0,1.5) (c1) [draw, dashed, red!50, rounded corners, minimum width=0.8cm, minimum height=2cm] {};
\node (rect) at (0,-1.5) (c1) [draw, dashed, red!50, rounded corners, minimum width=0.8cm, minimum height=2cm] {};
\node [ ] () at (0,3.2)  () {$\mathcal{H}_1$};
\node [ ] () at (0,-3.2)  () {$\mathcal{H}_2$};
\end{tikzpicture}
\end{center}
\caption{A 2-dimensional LTI system with two distinct, real, unstable eigenvalues (modes) $\lambda_1,\lambda_2$ is monitored by a network $\mathcal{G}$ of 10 nodes as shown above. Nodes 1-3 can detect $\lambda_1$, while nodes 8-10 can detect $\lambda_2$. Thus, the two minimal critical sets associated with the above system and network are $\mathcal{F}_1=\{1,2,3\}$ and $\mathcal{F}_2=\{8,9,10\}$. An example of a set that is critical, but not minimal, is $\{1,2,3,8\}$. The virtual source nodes associated with $\mathcal{F}_1$ and $\mathcal{F}_2$ are $s_1$ and $s_2$, respectively. There are no   $1$-total pair cuts w.r.t. $s_1$ or $s_2$. The set $\mathcal{H}=\{4,5,6,7\}$ is a $1$-local pair cut w.r.t. both $s_1$ and $s_2$ since $\mathcal{H}$ can be partitioned into $\mathcal{H}_1=\{4,5\}$ and $\mathcal{H}_2=\{6,7\}$, each of which are $1$-local sets. Since $\mathcal{H}_1$ and $\mathcal{H}_2$ are each $2$-total sets, $\mathcal{H}$ is also a $2$-total pair cut w.r.t. both $s_1$ and $s_2$.}
\label{fig:Illus}
\end{figure}
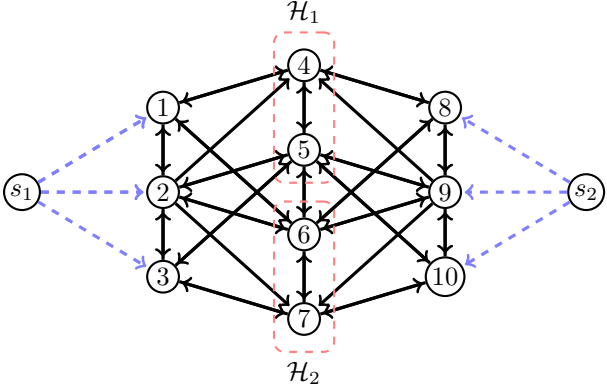

\begin{definition} (\textbf{Critical Set}) A set of nodes $\mathcal{F}\subset\mathcal{V}$ is said to be a critical set if the pair $(\mathbf{A},\mathbf{C}_{\mathcal{V}\setminus\mathcal{F}})$ is not detectable.
\end{definition}

Note that detectability of $(\mathbf{A},\mathbf{C})$ implies that a critical set must necessarily be non-empty.

\begin{definition} (\textbf{Minimal Critical Set}) A set $\mathcal{F} \subset \mathcal{V}$ is said to be a minimal critical set if $\mathcal{F}$ is a critical set and no subset of $\mathcal{F}$ is a critical set.
\end{definition}

Let $\mathcal{M}=\{\mathcal{F}_1, \cdots, \mathcal{F}_{|\mathcal{M}|}\}$ denote the set of all minimal critical sets. With each set $\mathcal{F}_i \in \mathcal{M}$, we associate a virtual node $s_i$ as follows. Directed edges are added from $s_i$ to each node in $\mathcal{F}_i$ and the resulting network is denoted by $\mathcal{G}^{'}_i=(\mathcal{V}\cup s_i, \mathcal{E}\cup\mathcal{E}_i)$, where $\mathcal{E}_i$ represents the set of edges from $s_i$ to $\mathcal{F}_i$. 

\begin{definition}($f$-\textbf{local pair and $f$-\textbf{total} pair cuts w.r.t.} ${s}_i$) Consider a minimal critical set $\mathcal{F}_i \in \mathcal{M}$. A set $\mathcal{H} \subset \mathcal{V}$ is called a cut w.r.t. $s_i$ if removal of $\mathcal{H}$ from $\mathcal{G}^{'}_i$ results in an induced subgraph of $\mathcal{G}^{'}_i$ whose node set can be partitioned into two non-empty sets $\mathcal{X}$ and $\mathcal{Y}$ with $s_i \in \mathcal{X}$, and no directed paths from $\mathcal{X}$ to $\mathcal{Y}$ in the induced subgraph. A cut $\mathcal{H}$ w.r.t. $s_i$ is called an $f$-local pair cut (resp., $f$-total pair cut) w.r.t. $s_i$ if it can be partitioned as $\mathcal{H}= \mathcal{H}_1 \cup \mathcal{H}_2$ such that both $\mathcal{H}_1$ and $\mathcal{H}_2$ are $f$-local (resp., $f$-total) in $\mathcal{G}$.
\label{defnt:cut}
\end{definition}

For an illustration of the above definitions, see Figure \ref{fig:Illus}. The following result identifies a fundamental limitation for $f$-local adversarial models.

\begin{thm} Suppose there exists an $f$-local pair cut w.r.t. $s_i$ in $\mathcal{G}^{'}_i$ for some minimal critical set $\mathcal{F}_i\in\mathcal{M}$. Then, it is impossible for any causal, synchronous and deterministic algorithm to solve Problem 1.
\label{thm:flocal}
\end{thm}
\vspace{-2mm}
\begin{proof} Suppose there exists an $f$-local pair cut $\mathcal{H}=\mathcal{H}_1\cup\mathcal{H}_2$ w.r.t. $s_i$ for some minimal critical set $\mathcal{F}_i \in \mathcal{M}$. For the sake of contradiction, suppose there exists a causal, synchronous and deterministic algorithm $\mathcal{T}$ that solves Problem 1 for the given network $\mathcal{G}$. From the definition of $\mathcal{H}$, we see that $\mathcal{Y}$ contains no elements of $\mathcal{F}_i$. Since $\mathcal{F}_i$ is a critical set, it then follows that the pair $(\mathbf{A},\mathbf{C}_{\mathcal{Y}})$ is not detectable. Thus, there exists an initial condition $\mathbf{x}[0]=\boldsymbol{\eta}$ that causes the measurement set $\mathbf{y}_{\mathcal{Y}}[k]$ corresponding to $\mathcal{Y}$ to be identically zero for all time, while the state $\mathbf{x}[k]$ remains bounded away from zero. The idea of the proof will be to demonstrate that the nodes in $\mathcal{Y}$ cannot distinguish between the zero initial condition and the initial condition $\boldsymbol{\eta}$ under an appropriately constructed attack. To this end, noting that each of the sets $\mathcal{H}_1$ and $\mathcal{H}_2$ are $f$-local and can hence act as valid adversarial sets, we consider the following executions $\sigma$ and $\sigma'$ of $\mathcal{T}$. 

\textbf{Execution} $\sigma$: The initial condition is $\mathbf{x}[0]=\mathbf{0}$. The nodes in $\mathcal{H}_1$ are regular while the nodes in $\mathcal{H}_2$ are adversarial. The nodes in $\mathcal{H}_2$ pretend that their state estimates are $\hat{\mathbf{x}}_{\mathcal{H}_2}[k]$ and that their measurements are $\mathbf{C}_{\mathcal{H}_2}\mathbf{A}^{k}\boldsymbol{\eta}$, where $\hat{\mathbf{x}}_{\mathcal{H}_2}[k]$ represents the collection of the state estimates maintained by the nodes in $\mathcal{H}_2$ during the execution $\sigma'$ of $\mathcal{T}$. Additionally, at each time-step, the nodes in $\mathcal{H}_2$ perform the exact same actions that they perform during the execution $\sigma'$.
 
\textbf{Execution} $\sigma'$: The initial condition is $\mathbf{x}[0]=\boldsymbol{\eta}$. The nodes in $\mathcal{H}_2$ are regular while the nodes in $\mathcal{H}_1$ are adversarial. The nodes in $\mathcal{H}_1$ pretend that their state estimates are $\hat{\mathbf{x}}_{\mathcal{H}_1}[k]$ and that their measurements are zero, where $\hat{\mathbf{x}}_{\mathcal{H}_1}[k]$ represents the collection of the state estimates maintained by the nodes in $\mathcal{H}_1$ during the execution $\sigma$ of $\mathcal{T}$. Additionally, at each time-step, the nodes in $\mathcal{H}_1$ perform the exact same actions that they perform during the execution $\sigma$.

Since the actions of the adversaries in the two executions described above are coupled, it becomes important to establish that such actions are in fact well-defined. To do so, we argue as follows. Consider the actions of the adversarial set $\mathcal{H}_2$ at time $k=0$ of execution $\sigma$. Due to their omniscient nature, these adversaries can anticipate the information that a regular set $\mathcal{H}_2$ is supposed to transmit at time $k=0$ of execution $\sigma'$ based on algorithm $\mathcal{T}$.  Thus, their actions are well-defined at time $k=0$. Note that the last two statements rely on the deterministic nature of $\mathcal{T}$. Specifically, under a deterministic algorithm $\mathcal{T}$, the actions of the regular nodes are also deterministic, and hence, can be predicted in advance by an omniscient adversary who is aware of the information set available to such regular nodes. An identical argument defines the actions of the adversarial set $\mathcal{H}_1$ at time $k=0$ of execution $\sigma'$. Since the actions of both the sets $\mathcal{H}_1$ and $\mathcal{H}_2$ at time $k=0$ are well-defined in each of the executions $\sigma$ and $\sigma'$, the response of the regular nodes to such actions (in the respective executions) at time $k=1$ can be anticipated by any adversarial set. Specifically, to generate their actions at time $k=1$ of execution $\sigma$ (resp., execution $\sigma'$), the adversarial set $\mathcal{H}_2$ (resp., $\mathcal{H}_1$) simply simulates execution $\sigma'$ (resp., execution $\sigma$) for time $k=0$ to figure out how a regular set $\mathcal{H}_2$ (resp., $\mathcal{H}_1$) would act at time $k=1$ of execution $\sigma'$ (resp., execution $\sigma$). Repeating the above argument reveals that the actions of the respective adversarial sets in each of the executions $\sigma$ and $\sigma'$ are well-defined at every time step.

Based on the attack described above, it is clear that the nodes in $\mathcal{Y}$ receive the same state estimate and measurement information from the nodes in $\mathcal{H}$ in each of the two executions. Further, their own measurements are identically zero for all time in each of the two executions. Hence, based on such identical information, it is impossible for the nodes in $\mathcal{Y}$ to resolve the difference in the underlying initial conditions via algorithm $\mathcal{T}$. This leads to the desired contradiction and completes the proof.
\end{proof}

\begin{remark} Interestingly, the necessary condition presented in the above theorem bears close resemblance to the necessary condition in \cite{flocal2} for resilient broadcasting subject to the same $f$-local Byzantine adversary model that we consider here. This similarity can be attributed to the following analogy: viewing the virtual nodes as originators of messages in a broadcasting context, Problem 1 can be interpreted as a version of the resilient broadcasting problem where the regular nodes are required to agree (asymptotically) on a time-varying message that captures the state evolution of the system. 
\end{remark}
 
Our next result provides a necessary condition for an $f$-total (and hence, also an $f$-local) adversarial model.

\begin{thm} Suppose there exists a {causal, synchronous and deterministic} algorithm that solves the variant of Problem $1$ corresponding to an $f$-total Byzantine adversary model. Then, the following equivalent statements are true.

\begin{itemize}
\item[(i)] Consider any minimal critical set $\mathcal{F}_i \in \mathcal{M}$. There exists no $f$-total pair cut w.r.t. $s_i$.
\item[(ii)] Consider a node $i \in \mathcal{V}$ such that $(\mathbf{A},\mathbf{C}_i)$ is not detectable. Let $\mathcal{X}_i$ denote the set of all nodes in $\mathcal{G}$ that have directed paths to node $i$, and consider a set $\mathcal{D}_i \subseteq \mathcal{X}_i$ such that $|\mathcal{D}_i| \leq 2f$. Let $\mathcal{P}_i \subseteq \mathcal{X}_i$ represent the set of nodes that have directed paths to node $i$ in the induced subgraph obtained by removing $\mathcal{D}_i$ from $\mathcal{G}$. Then, $(\mathbf{A},\mathbf{C}_{i \cup \mathcal{P}_i})$ is detectable.
\end{itemize}
\label{thm:ftotal}
\end{thm}

The proof of necessity mimics the proof of Theorem \ref{thm:flocal}, while the equivalence between the two conditions stated in Theorem \ref{thm:ftotal} is established in Appendix \ref{app:proofthmftotal}.

\begin{remark} 
In \cite{fawzi,joao2}, the authors showed that for centralized systems subject to $f$ sensor attacks, a necessary condition for estimating the state asymptotically is that the system should remain detectable after the removal of any $2f$ sensors. In our present distributed setting, the maximum information about the state that any given node $i$ can hope to obtain is from the set $\{i\cup\mathcal{X}_i\}$, where $\mathcal{X}_i$ is defined as in Theorem \ref{thm:ftotal}. Thus, the second part of Theorem \ref{thm:ftotal} generalizes the necessary conditions in \cite{fawzi,joao2}. In \cite{sundaramtac,fabio2}, the authors established that the graph-connectivity metric plays a pivotal role in the analysis of fault-tolerant and resilient distributed consensus algorithms for settings where there are no underlying state dynamics that need to be estimated. The results stated in Theorems \ref{thm:flocal} and \ref{thm:ftotal} differ from those in \cite{sundaramtac,fabio2} since they blend both graph-theoretic and system-theoretic requirements. Finally, it can be easily shown that when there are no adversaries, i.e., when $f=0$, the conditions identified in Theorem \ref{thm:ftotal} reduce to the necessary and sufficient condition for distributed state estimation, namely every source component (strong components with no incoming edges) of the graph should be detectable \cite{martins3,allerton,mitraarchive,wang,han}.
\end{remark} 

We now discuss certain implications of Theorem \ref{thm:ftotal}.
Given an LTI system \eqref{eqn:plant}, a measurement model specified by \eqref{eqn:Obsmodel}, and a communication graph $\mathcal{G}$, it is of both theoretical and practical interest to know the maximum number of adversaries that can be tolerated when one seeks to solve Problem 1. Leveraging Theorem \ref{thm:ftotal}, we can provide an upper bound on this number, as follows.

\begin{corollary}
Let $k$ denote the smallest positive integer such that there exists a $k$-total pair cut w.r.t. $s_i$ for some $\mathcal{F}_i\in\mathcal{M}$. Then, the total number of adversaries $f$ must satisfy the inequality $f < k$ for Problem 1 to have a solution.\footnote{Similar bounds for static power system models subject to attacks were obtained in \cite{Kosut}.}
\label{corr:upperbound}
\end{corollary}

\begin{corollary} 
The condition $|\mathcal{F}_i| \geq (2f+1)\ \forall \mathcal{F}_i \in \mathcal{M}$ is necessary for resilient distributed state estimation subject to the $f$-local or $f$-total adversarial model.
\label{corr:card}
\end{corollary}
The proof of the above result is straightforward and is hence omitted here.
With the above corollary in hand, one can gain insights regarding the distribution of certain specific critical sets in the network. To do so, given an eigenvalue $\lambda_j \in \Lambda_{U}(\mathbf{A})$, let $\{\boldsymbol{\rho}^{(j)}_1, \cdots, \boldsymbol{\rho}^{(j)}_{g_{\mathbf{A}}(\lambda_j)}\}$ represent a basis for the null space of $(\mathbf{A}-\lambda_j\mathbf{I}_n)$, and let $\phi^{(j)}_i=span\{\boldsymbol{\rho}^{(j)}_i\}, i\in \{1,\cdots,g_{\mathbf{A}}(\lambda_j)\}$. We say that node $i$ can detect the subspace $\phi^{(j)}_i$  if $\mathbf{C}_{i}\boldsymbol{\rho}^{(j)}_i \neq \mathbf{0}$.\footnote{Throughout the paper, for the sake of conciseness, we use the terminology ``node $i$ can detect eigenvalue $\lambda_j$" to imply that $rank\left[\begin{smallmatrix} \mathbf{A}-\lambda_j\mathbf{I}_n \\ \mathbf{C}_i \end{smallmatrix}\right] = n$. Each stable eigenvalue is considered detectable w.r.t. the measurements of every node.} Let $\mathcal{W}^{(j)}_i \subseteq \mathcal{V}$ denote the set of all nodes that can detect $\phi^{(j)}_i$. The next result then readily follows from Corollary \ref{corr:card} and the classical PBH test \cite{hespanha}.

\begin{proposition}
For each $\lambda_j \in \Lambda_{U}(\mathbf{A})$, if $\mathcal{W}^{(j)}_i \subset \mathcal{V}$, where $\ 1\leq i \leq g_{\mathbf{A}}(\lambda_j)$, then  $|\mathcal{W}^{(j)}_i| \geq (2f+1)$ is a necessary condition for resilient distributed state estimation subject to the f-local or f-total adversarial model.
\label{prop:fund}
\end{proposition}
{For systems with distinct eigenvalues, a direct consequence of the above result is the requirement of at least $(2f+1)$ nodes that can detect each unstable eigenvalue of the system}. The preceding analysis builds up to the distributed estimation strategy adopted in  this paper. In particular, our approach involves identifying the locally detectable and undetectable eigenvalues associated with a given node, and subsequently devising separate estimation strategies for the subspaces associated with such eigenvalues. We formalize this idea in the next section.

\begin{remark} Two important directions of future investigation are (i) finding an efficient algorithm (if one exists) for computing $k$ in Corollary \ref{corr:upperbound} either exactly or approximately, and (ii) determining whether the conditions stated in Theorem 1 (resp., Theorem 2) are sufficient for achieving resilient distributed state estimation subject to an $f$-local (resp., $f$-total) Byzantine adversary model. Note that the main source of computational complexity associated with the first point lies in finding all the minimal critical sets associated with the given system.
\label{rem:directions}
\end{remark}
\section{Resilient Distributed State Estimation}
\label{section:estimation}
\subsection{Preliminaries}
\label{sec:prelim}
For each eigenvalue $\lambda\in sp(\mathbf{A})$, let $\mathbf{V}(\lambda)$ represent a block diagonal matrix with the Jordan blocks corresponding to $\lambda$ (in the standard Jordan canonical representation of $\mathbf{A}$) along the main block diagonal. We begin by recalling certain properties of the real Jordan canonical form of a square matrix that will be useful for our subsequent development \cite{horn}. We first note that if $\lambda$ represents a non-real eigenvalue of $\mathbf{A}$ and $\bar{\lambda}$ represents its complex-conjugate, then \cite[Lemma 3.1.18]{horn} ensures that $\lambda$ and $\bar{\lambda}$ have the same Jordan structure. Next, let $\lambda=a+ib$ where $a,b \in \mathbb{R}$, and $i=\sqrt{-1}$. Let $\mathbf{D}(a,b)$ be defined as $\mathbf{D}(a,b)\triangleq \begin{bmatrix}a&b\\-b&a\end{bmatrix}$. Then, the matrix $diag(\mathbf{V}(\lambda),\mathbf{V}(\bar{\lambda}))$ is similar to a real block upper triangular matrix $\mathbf{W}(\lambda)\in{\mathbb{R}}^{2a_{\mathbf{A}}(\lambda)\times 2a_{\mathbf{A}}(\lambda)}$ which has $a_{\mathbf{A}}(\lambda)$ 2-by-2 blocks $\mathbf{D}(a,b)$ on the main block diagonal and $(a_{\mathbf{A}}(\lambda)-1)$ blocks $\mathbf{I}_2$ on the block superdiagonal.  Henceforth, for a non-real eigenvalue $\lambda \in sp(\mathbf{A})$, $\mathbf{W}(\lambda)$ will have the meaning  discussed above. Let $sp(\mathbf{A})=\{\{\lambda_1,\bar{\lambda}_1\}, \cdots, \{\lambda_p,\bar{\lambda}_p\}, \lambda_{p+1}, \cdots, \lambda_{\gamma}\}$ with the first $p$ pairs representing the non-real eigenvalues, and $\lambda_{p+1}$ to $\lambda_{\gamma}$ representing the real eigenvalues of $\mathbf{A}$. Then, the \textit{real Jordan canonical form theorem} \cite[Theorem 3.4.1.5]{horn} can be stated as follows.
\begin{thm} There exists a real similarity transformation matrix $\mathbf{T}$ that transforms the state transition matrix $\mathbf{A}$ in \eqref{eqn:plant} to a real block diagonal matrix $\mathbf{M}$ given by $\mathbf{M}=diag(\mathbf{W}(\lambda_1), \cdots, \mathbf{W}(\lambda_p), \mathbf{V}(\lambda_{p+1}), \cdots, \mathbf{V}(\lambda_{\gamma}))$.
\label{thm:RealJordan}
\end{thm}

With $\mathbf{T}$ as in the above theorem, and $\mathbf{z}[k]=\mathbf{T}^{-1}\mathbf{x}[k]$, the dynamics \eqref{eqn:plant} are  transformed into the form
\begin{equation}
\begin{split}
\mathbf{z}[k+1] &= \mathbf{Mz}[k] \\
\mathbf{y}_i[k] &= \bar{\mathbf{C}}_i\mathbf{z}[k], \quad \forall i \in \{1, \cdots, N\} \enspace 
\end{split}
\label{eqn:plant_tr}
\end{equation}
where $\mathbf{M}=\mathbf{{T}^{-1}AT}$ and $\bar{\mathbf{C}}_i=\mathbf{C}_i\mathbf{T}$.
%Next, let $\mathbf{J}_{k}(\lambda)$ represent a Jordan block of order $k$ corresponding to the non-real eigenvalue $\lambda\in sp(\mathbf{A})$ in the standard Jordan canonical representation of $\mathbf{A}$. Let $\lambda=a+ib$, where $a,b \in \mathbb{R}$. Then, the matrix $diag(\mathbf{J}_{k}(\lambda),\mathbf{J}_{k}(\bar{\lambda}))$ is permutation similar to the following block upper triangular matrix of order $2k$
%\begin{equation}
%\mathbf{D}_k(a,b)=\begin{bmatrix}\mathbf{D}(a,b)&\mathbf{I}_{2}&{}&{} \\
%{}&\ddots &\ddots &{}\\
%{}&{}&\mathbf{D}(a,b)&\mathbf{I}_{2}\\
%{}&{}&{}&\mathbf{D}(a,b)\\
%\end{bmatrix},
%\label{eq:realjordan}
%\end{equation}
For a non-real eigenvalue pair $\{\lambda_j,\bar{\lambda}_j\} \in sp(\mathbf{A})$, let $\mathbf{z}^{(j)}[k] \in \mathbb{R}^{2a_{\mathbf{A}}(\lambda_j)}$ represent the portion of the state $\mathbf{z}[k]$ associated with the matrix $\mathbf{W}(\lambda_j)$. Similarly, for a real eigenvalue $\lambda_j\in sp(\mathbf{A})$, $\mathbf{z}^{(j)}[k] \in \mathbb{R}^{a_{\mathbf{A}}(\lambda_j)}$ is the portion of the state $\mathbf{z}[k]$ associated with the matrix $\mathbf{V}(\lambda_j)$. For each node $i$, we denote the detectable and undetectable eigenvalues by the sets $\mathcal{O}_{i}$ and {$\overline{\mathcal{O}}_i$}, respectively. Next, we introduce the notion of \textit{source nodes}.

\begin{definition}
(\textbf{Source nodes})
For each $\lambda_j \in \Lambda_{U}(\mathbf{A})$, the set of nodes that can detect $\lambda_j$ is denoted by $\mathcal{S}_j$, and called the set of \textit{source nodes} for $\lambda_j$.
\end{definition}

We now proceed to develop an estimation scheme that enables each regular node to estimate $\mathbf{z}[k]$ (from which they can obtain $\mathbf{x}[k]=\mathbf{Tz}[k]$). Accordingly, let $\hat{\mathbf{z}}^{(j)}_i[k]$ denote the estimate of $\mathbf{z}^{(j)}[k]$ (the portion of $\mathbf{z}[k]$ corresponding to the eigenvalue $\lambda_j$\footnote{Throughout the rest of the paper, the terminology ``$\mathbf{z}^{(j)}[k]$ corresponds to the eigenvalue $\lambda_j$" should be interpreted as $\mathbf{z}^{(j)}[k]$ corresponds to the eigenvalue pair $\{\lambda_j,\bar{\lambda}_j\}$ for a non-real eigenvalue $\lambda_j\in sp(\mathbf{A}).$}) maintained by node $i\in\mathcal{R}$. For each $\lambda_j\in \Lambda_{U}(\mathbf{A})$, our estimation scheme relies on separate strategies for nodes in $\mathcal{S}_j$ and $\mathcal{V}\setminus\mathcal{S}_j$. In particular, each node in $\mathcal{S}_j$ employs a Luenberger observer for estimating $\mathbf{z}^{(j)}[k]$. The nodes in $\mathcal{V}\setminus\mathcal{S}_j$, on the other hand, cannot detect the eigenvalue $\lambda_j$, and thus rely on a resilient consensus algorithm to estimate $\mathbf{z}^{(j)}[k]$. In what follows, we discuss these ideas in detail.

The first step in the estimation process involves the above common coordinate transformation given by $\mathbf{z}[k]=\mathbf{T}^{-1}\mathbf{x}[k]$, to be performed by each regular node of the graph. As this only requires knowledge of the system matrix $\mathbf{A}$ (which is assumed to be known by all the nodes), all of the nodes can do this in a distributed manner (e.g., by using an agreed-upon convention for ordering the eigenvalues and corresponding eigenvectors). Building on the general theme in \cite{allerton}, we first present a method that allows a regular node $i\in\mathcal{R}$ to estimate the locally detectable portion of the state $\mathbf{z}[k]$ \textit{without} communicating with neighbors. To this end, consider the following result. 

\begin{lemma}
Let $\lambda_j \in \mathcal{O}_i$ be a non-real eigenvalue. Let $\bar{\mathbf{C}}^{(j)}_i$ denote the columns of $\bar{\mathbf{C}}_i$ corresponding to $\mathbf{W}(\lambda_j)$ in \eqref{eqn:plant_tr}. Then, the pair $(\mathbf{W}(\lambda_j),\bar{\mathbf{C}}^{(j)}_i)$ is detectable.
\label{lemma:compconj}
\end{lemma}
\begin{proof} We claim that $\lambda_j \in \mathcal{O}_i$ if and only if $\bar{\lambda}_j \in \mathcal{O}_i$. It suffices to prove necessity since the proof for sufficiency will follow an identical argument. We prove necessity by contradiction. Suppose node $i$ can detect $\lambda_j$ (i.e., $\lambda_j \in \mathcal{O}_i$), but cannot detect $\bar{\lambda}_j$. Then, there exists $\mathbf{v} \neq \mathbf{0}$ such that $\mathbf{Av}=\bar{\lambda}_j\mathbf{v}$ and $\mathbf{C}_i\mathbf{v}=\mathbf{0}$. Taking complex conjugates on both sides of these equations reveals that node $i$ cannot detect $\lambda_j$, leading to the desired contradiction. Given that a similarity transformation maps $(\mathbf{A,C}_i)$ to $(\mathbf{M,\bar{C}}_i)$, it then follows that $\{\lambda_j,\bar{\lambda}_j\}$ are detectable eigenvalues w.r.t. $(\mathbf{M,\bar{C}}_i)$. Detectability of the pair $(\mathbf{W}(\lambda_j),\bar{\mathbf{C}}^{(j)}_i)$ then follows readily from the PBH test by noting the structure of the matrix $\mathbf{M}$.
\end{proof}

Let $\mathcal{O}_i=\{\{\lambda_{n_1},\bar{\lambda}_{n_1}\}, \cdots, \{\lambda_{n_{p_i}},\bar{\lambda}_{n_{p_i}}\}, \lambda_{n_{p_i+1}}, \cdots, \lambda_{n_{\gamma_i}}\}$, where the first $p_i$ pairs represent the non-real eigenvalues, and $\lambda_{n_{p_i+1}}$ to $\lambda_{n_{\gamma_i}}$ represent the real eigenvalues of $\mathbf{A}$ that are detectable w.r.t. the measurements of node $i$. Let $\mathbf{M}_{\mathcal{O}_i} = diag(\mathbf{W}(\lambda_{n_1}), \cdots, \mathbf{W}(\lambda_{n_{p_i}}), \mathbf{V}(\lambda_{n_{p_i+1}}),$ $\cdots,
 \mathbf{V}(\lambda_{n_{\gamma_i}}))$. Let $\mathbf{C}_{\mathcal{O}_i}$ represent the columns of $\bar{\mathbf{C}}_{i}$ corresponding to the matrix $\mathbf{M}_{\mathcal{O}_i}$, and $\mathbf{z}_{\mathcal{O}_i}[k]$ denote the portion of the state $\mathbf{z}[k]$ corresponding to the detectable eigenvalues of node $i$, i.e., corresponding to $\mathcal{O}_i$. Based on Lemma \ref{lemma:compconj}, it is easy to see that the pair $(\mathbf{M}_{\mathcal{O}_i}, \mathbf{C}_{\mathcal{O}_i})$ is detectable. Thus, a standard Luenberger observer can be locally constructed by node $i$ for estimating $\mathbf{z}_{\mathcal{O}_i}[k]$. The {details} of such a construction are straightforward, and are similar to those in Section VI-A of \cite{mitraarchive}. We thus skip minor details and state the following result which will be useful later on.
 
\begin{lemma}
For each regular node $i \in \mathcal{R}$ and each $\lambda_j \in \mathcal{O}_i$, a Luenberger observer can be locally constructed by node $i$ such that $\lim_{k\to\infty}||\hat{\mathbf{z}}^{(j)}_{i}[k]-\mathbf{z}^{(j)}[k]||=0$.
\label{lemma:luen}
\end{lemma}

Based on the previous result, we see that a regular node $i$ can estimate certain portions of the state space without having to exchange information with neighbors. The challenge, however, lies in estimating the locally undetectable portion of the state in the presence of adversaries. The following section presents a resilient consensus based strategy to address this issue.

\subsection{Local-Filtering Based Resilient Estimation}
\label{sec:LFSE}
 For any $\lambda_j \in sp(\mathbf{A})$, let $z^{(jm)}[k]$ denote the $m$-th component of the vector $\mathbf{z}^{(j)}[k]$, and let $\hat{z}^{(jm)}_i[k]$ denote the estimate of that component maintained by node $i\in\mathcal{V}$. Consider an unstable eigenvalue $\lambda_j \in $ {$\overline{\mathcal{O}}_i$}. For such an eigenvalue, node $i$ has to rely on the information received from its neighbors, {some of which} might be adversarial, in order to estimate $\mathbf{z}^{(j)}[k]$. To this end, we propose a  resilient consensus algorithm that requires each regular node $i \in \mathcal{V} \setminus \mathcal{S}_j$ to update its estimate of $\mathbf{z}^{(j)}[k]$ using the following two stage filtering strategy:
\begin{itemize}
\item[1)] At each time-step $k$, each regular node $i$ collects the state estimates of $\mathbf{z}^{(j)}[k]$ received from \textit{only} those neighbors that belong to a certain subset $\mathcal{N}^{(j)}_i \subseteq \mathcal{N}_i$ (to be defined later). For every component $m$ of $\mathbf{z}^{(j)}[k]$, the estimates of $z^{(jm)}[k]$ received from nodes in $\mathcal{N}^{(j)}_i$ are sorted from largest to smallest.
\item[2)] For each component $m$ of $\mathbf{z}^{(j)}[k]$, node $i$ removes the largest and smallest $f$ estimates (i.e., removes $2f$ estimates in all) of $z^{(jm)}[k]$ received from nodes in $\mathcal{N}^{(j)}_i$, and computes the quantity:
\begin{equation}
\bar{z}^{(jm)}_i[k]=\sum_{l \in \mathcal{M}^{(jm)}_i[k]}w^{(jm)}_{il}[k]\hat{z}^{(jm)}_{l}[k],
\label{eqn:update rule1}
\end{equation}
where $\mathcal{M}^{(jm)}_i[k] \subset \mathcal{N}^{(j)}_i (\subseteq \mathcal{N}_i)$ is the set of nodes from which node $i$ chooses to accept estimates of $z^{(jm)}[k]$ at time-step $k$, after removing the $f$ largest and $f$ smallest estimates of ${z}^{(jm)}[k]$ from $\mathcal{N}^{(j)}_i$. Node $i$ assigns the weight $w^{(jm)}_{il}[k]$ to the $l$-th node at the $k$-th time-step for estimating the $m$-th component of $\mathbf{z}^{(j)}[k]$. The weights are nonnegative and chosen to satisfy $\sum_{l \in \mathcal{M}^{(jm)}_i[k]}{w^{(jm)}_{il}}[k]=1, \forall \lambda_j \in$ {$\overline{\mathcal{O}}_i$} and for each component $m$ of $\mathbf{z}^{(j)}[k]$. With the quantities $\bar{z}^{(jm)}_i[k]$ in hand, node $i$ updates $\hat{\mathbf{z}}^{(j)}_i[k]$ as follows:
\begin{equation}
\hat{\mathbf{z}}^{(j)}_i[k+1]=
    \begin{cases}
    \mathbf{V}(\lambda_j)\bar{\mathbf{z}}^{(j)}_i[k], & \text{if}\ \lambda_j \in {\overline{\mathcal{O}}_i} \color{black}{\ \text{is real}} \\
      \mathbf{W}(\lambda_j)\bar{\mathbf{z}}^{(j)}_i[k], & \text{if}\ \lambda_j \in {\overline{\mathcal{O}}_i} \color{black}{\ \text{is not real}},
    \end{cases}
\label{eq:update2}
\end{equation}
where $\bar{\mathbf{z}}^{(j)}_i[k]={\begin{bmatrix} \bar{z}^{(j1)}_{i}[k], \cdots, \bar{z}^{(j\sigma_j)}_i[k]\end{bmatrix}}^{T}$, $\sigma_j=a_{\mathbf{A}}(\lambda_j)$ if $\lambda_j\in \overline{\mathcal{O}}_i$ is real, and $\sigma_j=2a_{\mathbf{A}}(\lambda_j)$ if $\lambda_j\in \overline{\mathcal{O}}_i$ is not real.
\end{itemize}
We refer to the above algorithm as the Local-Filtering based Resilient Estimation (LFRE) algorithm. For implementing this algorithm, a regular node $i$ needs to construct the set $\mathcal{N}^{(j)}_i$, $\forall\lambda_j \in$ {$\overline{\mathcal{O}}_i$}, based on the relative positions of its neighbors (with respect to its own position) in  $\mathcal{G}$. We will provide the exact definition of $\mathcal{N}^{(j)}_i$, and a distributed algorithm for constructing such a set in the following sections where we analyze the convergence of the LFRE algorithm. We conclude this section by commenting on certain features of the LFRE algorithm.

\begin{remark}
The rationale behind performing a real Jordan canonical decomposition at every node (as opposed to a standard Jordan transformation) is to ensure that the state estimates featuring in equations \eqref{eqn:update rule1} and \eqref{eq:update2} are real at every time-step, thereby making the sorting operation performed in Step 1 of the algorithm meaningful. At any time-step, if a regular node $i$ either receives a non-real estimate of $z^{(jm)}[k]$ from some node $l \in \mathcal{N}^{(j)}_i$ or does not receive an estimate at all, it would immediately identify node $l$ as an adversarial node, and simply assign a $0$ value to node $l$'s estimate of $z^{(jm)}[k]$. Note that every regular node in $\mathcal{N}^{(j)}_i$ will always transmit a real estimate to node $i$ at every time-step.
\end{remark}

\begin{remark}
The strategy of disregarding the most extreme values in one's neighborhood, and using a convex combination of the rest for performing linear scalar updates, has been used for designing resilient distributed algorithms for consensus \cite{WMSR,rescons,vaidyacons} and optimization \cite{Sundaramopt,su} problems. In this paper, we show that such algorithms can also be used for resilient distributed state estimation, with certain substantial differences arising from the fact that the nodes are trying to track the state of an external dynamical system.
\end{remark}

{\begin{remark} The consensus weights $w^{(jm)}_{il}$ appearing in equation \eqref{eqn:update rule1} can be chosen arbitrarily to achieve an exponential rate of convergence, as long as the weights meet the rules specified by the LFRE algorithm. Since our primary focus is on resilience against worst-case adversarial behavior, the problem of optimizing such weights (or exploiting sensor memory) for achieving improved performance against noise is not considered in this paper. In a non-adversarial setting (i.e., when $f=0$), the proposed LFRE algorithm will continue to guarantee exponential convergence in the absence of noise, and bounded mean square error in the presence of i.i.d. noise with bounded second moments (provided the topological conditions outlined in Section \ref{sec:feasible} are met). However, disregarding the estimates of certain neighbors in the absence of attacks may potentially degrade performance against noise; we do not delve deeper into this topic here.
\end{remark}}

{It should be noted that the algorithmic development in this section can be considerably simplified if more structure is imposed on the system matrix $\mathbf{A}$ (for instance, the assumption made in \cite{mitraCDC} that $\mathbf{A}$ has only real, distinct eigenvalues).}
\section{Analysis of the Resilient Distributed Estimation Strategy}
\label{sec:analysis}
In this section, we provide our main result concerning the convergence of the LFRE algorithm. Let $\Omega_{U}(\mathbf{A}) \triangleq \{\lambda_j \in \Lambda_{U}(\mathbf{A})| \mathcal{V}\setminus\mathcal{S}_j \, \textrm{is non-empty} \}$. By this definition, all nodes are source nodes for each eigenvalue in $\Lambda_{U}(\mathbf{A}) \setminus \Omega_{U}(\mathbf{A})$, and are hence capable of recovering the corresponding portions of the state based on locally constructed Luenberger observers (as discussed in Section \ref{sec:prelim}). Consequently, the LFRE algorithm specifically applies to only those eigenvalues that belong to $\Omega_{U}(\mathbf{A})$. Consider the following definition.

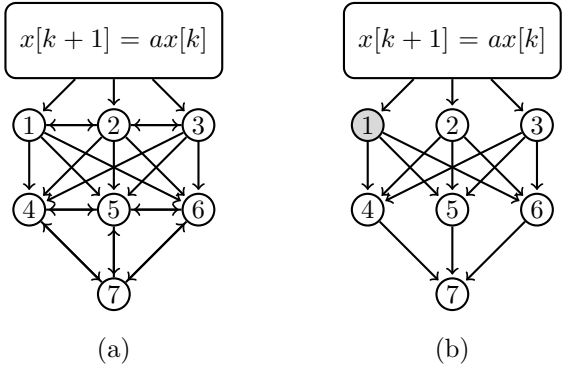
\begin{figure}[t]
\begin{center}
\begin{tikzpicture}
[->,shorten >=1pt,scale=.75,inner sep=1pt, minimum size=12pt, auto=center, node distance=3cm,
  thick, node/.style={circle, draw=black, thick},]
\tikzstyle{block} = [rectangle, draw, fill=white, 
    text width=5em, text centered, rounded corners, minimum height=4em];
\tikzstyle{block} = [rectangle, draw, fill=white, 
    text width=8em, text centered, rounded corners, minimum height=1cm, minimum width=1cm];
 \node [block]  at (0,9) (plant) {$x[k+1]=ax[k]$};
\node [circle, draw, fill=white](n1) at (-1.5,7.5)  (1)  {1};
\node [circle, draw, fill=white](n2) at (0,7.5)   (2)  {2};
\node [circle, draw, fill=white](n3) at (1.5,7.5)   (3)  {3};
\node [circle, draw, fill=white](n4) at (-1.5,6)  (4)  {4};
\node [circle, draw, fill=white](n5) at (0,6)   (5)  {5};
\node [circle, draw, fill=white] (n6) at (1.5,6)   (6)  {6};
\node [circle, draw, fill=white](n7) at (0,4.5)   (7)  {7};
\node [ ] (n8) at (0,3.5)  (8) {(a)};

\path[every node/.style={font=\sffamily\small}]

     (plant)
         edge [right] node [] {}(1)
         edge [right] node [] {}(2)
         edge [right] node [] {}(3)

     (1)
         edge [right] node [] {}(2)
     
         edge [right] node [] {}(4)
         edge [right] node [] {}(5)
         edge [right] node [] {}(6)
         
     (2) 
         edge [right] node [] {}(1)
         edge [right] node [] {}(3)
         edge [right] node [] {}(4)
         edge [right] node [] {}(5)
         edge [right] node [] {}(6)
         
     (3) 
     edge [right] node [] {}(2)
         edge [right] node [] {}(4)
         edge [right] node [] {}(5)
         edge [right] node [] {}(6)
         
 (4)     edge [right] node [] {}(5)
         edge [right] node [] {}(7)
         
 (5)     edge [right] node [] {}(4)
         edge [right] node [] {}(6)
         edge [right] node [] {}(7)
         
 (6)     
         edge [right] node [] {}(5)
         edge [right] node [] {}(7)
         
(7)
         
       edge [right] node [] {}(4)
       edge [right] node [] {}(5)
       edge [right] node [] {}(6);        
            
\node [block]  at (6,9) (plant) {$x[k+1]=ax[k]$};          
\node [circle, draw, fill=gray!30](n1) at (4.5,7.5)  (1)  {1};
\node [circle, draw, fill=white](n2) at (6,7.5)   (2)  {2};
\node [circle, draw, fill=white](n3) at (7.5,7.5)   (3)  {3};
\node [circle, draw, fill=white](n4) at (4.5,6)  (4)  {4};
\node [circle, draw, fill=white](n5) at (6,6)   (5)  {5};
\node [circle, draw, fill=white] (n6) at (7.5,6)   (6)  {6};
\node [circle, draw, fill=white](n7) at (6,4.5)   (7)  {7};
\node [ ] (n8) at (6,3.5)  (10) {(b)};
\path[every node/.style={font=\sffamily\small}]
     (plant)
         edge [right] node [] {}(1)
         edge [right] node [] {}(2)
         edge [right] node [] {}(3)

     (1)
         edge [right] node [] {}(4)
         edge [right] node [] {}(5)
         edge [right] node [] {}(6)
         
     (2)
         edge [right] node [] {}(4)
         edge [right] node [] {}(5)
         edge [right] node [] {}(6)
         
     (3) 
         edge [right] node [] {}(4)
         edge [right] node [] {}(5)
         edge [right] node [] {}(6)
         
 (4) 
         edge [right] node [] {}(7)
         
 (5) 
         edge [right] node [] {}(7)
         
 (6) 
         edge [right] node [] {}(7); 
            
\end{tikzpicture}
\end{center}
\caption{A scalar unstable plant is monitored by a network of 7 nodes as depicted by the figure on the left. Nodes 1, 2 and 3 are the source nodes for this system. The figure on the right represents a subgraph of the original graph satisfying the properties of a MEDAG in Definition \ref{defn:MEDAG} for all $1$-local sets (i.e., with $f=1$). For example, when $\mathcal{A} = \{1\}$ (as shown in the right figure), every non-source node has at least $2f+1 = 3$ neighbors. The levels that partition $\mathcal{R}=\mathcal{V}\setminus\mathcal{A}$ are level 0 with nodes 2 and 3, level 1 with nodes 4, 5 and 6, and level 2 with node 7. Each regular node has all its regular neighbors in levels that are numbered lower than its own.}
\label{fig:sim_DAG}
\end{figure}

\begin{definition}\label{defn:MEDAG}
(\textbf{Mode Estimation Directed Acyclic Graph (MEDAG)}) Consider an eigenvalue $\lambda_j \in \Omega_{U}(\mathbf{A})$. Suppose there exists a spanning subgraph $\mathcal{G}_j = (\mathcal{V},\mathcal{E}_j)$ of $\mathcal{G}$  with the following properties for all $f$-local sets $\mathcal{A}$ and $\mathcal{R}=\mathcal{V}\setminus\mathcal{A}$.
\begin{itemize}
\item[(i)] If $i \in \{\mathcal{V}\setminus\mathcal{S}_j\} \cap \mathcal{R}$, then $|\mathcal{N}^{(j)}_i| \geq 2f+1$, where $\mathcal{N}^{(j)}_i=\{l|(l,i) \in \mathcal{E}_j\}$ represents the neighborhood of node $i$ in $\mathcal{G}_j$.
\item[(ii)] There exists a partition of $\mathcal{R}$ into the sets $\{\mathcal{L}^{(j)}_{0}, \cdots , \mathcal{L}^{(j)}_{T_j}\}$, where $T_j \in \mathbb{N}_{+}$, $\mathcal{L}^{(j)}_{0} = \mathcal{S}_j \cap \mathcal{R}$, and if $i \in \mathcal{L}^{(j)}_q$ (where $1 \leq q \leq T_j)$, then $\mathcal{N}^{(j)}_i \cap \mathcal{R} \subseteq \bigcup^{q-1}_{r=0} \mathcal{L}^{(j)}_r$. Furthermore, $\mathcal{N}^{(j)}_i=\emptyset, \forall i \in \mathcal{L}^{(j)}_0$.
\end{itemize}
Then, we call $\mathcal{G}_j$ a \textit{Mode Estimation Directed Acyclic Graph (MEDAG)} for $\lambda_j \in \Omega_{U}(\mathbf{A})$.
\end{definition}
An example of a MEDAG is shown in Figure \ref{fig:sim_DAG}.
The ``for all $\mathcal{A}$" in the definition accounts for the fact that the set of adversarial nodes during the process of state estimation is unknown, and hence can be any $f$-local set of $\mathcal{V}$. Note that $T_j$ and the levels $\mathcal{L}^{(j)}_0$ to $\mathcal{L}^{(j)}_{T_j}$ can vary across different $f$-local sets. For a given $f$-local set $\mathcal{A}$, we say a regular node $i\in\mathcal{L}^{(j)}_m$ ``belongs to level $m$", where the levels are indicative of the distances of the regular nodes from the source set $\mathcal{S}_j$. The first property indicates that every regular node $i \in \mathcal{V}\setminus{\mathcal{S}_j}$ has at least $(2f+1)$ neighbors in the subgraph $\mathcal{G}_j$, while the second property indicates that all its regular neighbors in such a subgraph belong to levels strictly preceding its own level. In essence, the edges of the MEDAG $\mathcal{G}_j$ represent a medium for transmitting information securely from the source nodes $\mathcal{S}_j$ to the non-source nodes, by preventing the adversaries from forming a bottleneck between such nodes. Intuitively, this requires redundant nodes and edges, and such a requirement is met by the first property of the MEDAG. In particular, as regards measurement redundancy, it follows from the definition that for each $\lambda_j \in \Omega_{U}(\mathbf{A})$, a MEDAG $\mathcal{G}_j$ contains at least $(2f+1)$ source nodes that can detect $\lambda_j$.\footnote{Recall from the discussion immediately following Proposition \ref{prop:fund} that such a condition is in fact necessary for systems with distinct eigenvalues.} The LFRE algorithm described in the previous section relies on a special \textit{uni-directional} information flow pattern that requires a node $i$ to listen to \textit{only} its neighbors in $\mathcal{N}^{(j)}_i$ for estimating $\mathbf{z}^{(j)}[k]$. The second property of a MEDAG then indicates that nodes in level $m$ only use the estimates of regular nodes in levels $0$ to $m-1$ for recovering $\mathbf{z}^{(j)}[k]$. The implications of the above properties will become apparent in the proof of the following result which provides a sufficient condition for solving Problem 1 based on our approach.

\begin{thm} Suppose that $\mathcal{G}$ contains a MEDAG $\mathcal{G}_j$ for each $\lambda_j \in \Omega_{U}(\mathbf{A})$. Then, based on the LFRE dynamics described by equations \eqref{eqn:update rule1} and \eqref{eq:update2}, each regular node $i\in\mathcal{R}$ can asymptotically estimate the state of the plant, despite the actions of any $f$-locally bounded set of Byzantine adversaries.
\label{thm:consensus}
\end{thm}
The proof of the above theorem is given in Appendix \ref{app:proofthm4}. Notice that Theorem \ref{thm:consensus} hinges on the existence of a MEDAG $\mathcal{G}_j$, for each $\lambda_j \in \Omega_{U}(\mathbf{A})$; in the following section we describe an approach for checking whether a given graph $\mathcal{G}$ contains such MEDAGs.
\section{Checking the Existence of a MEDAG}
\label{sec:MEDAG} 
From the foregoing discussion, it is apparent that the MEDAGs described in Definition \ref{defn:MEDAG} play a key role in solving Problem 1 based on our proposed technique. In particular, recall that for each $\lambda_j\in\Omega_{U}(\mathbf{A})$, the LFRE algorithm described in Section \ref{sec:LFSE} requires a regular node $i \in \mathcal{V}\setminus\mathcal{S}_j$ to accept estimates from only its neighbor set $\mathcal{N}^{(j)}_i$ in the MEDAG $\mathcal{G}_j$ for estimating $\mathbf{z}^{(j)}[k]$. With these points in mind, our immediate goal in this section will be to develop a distributed algorithm, namely Algorithm 1, that constructs a MEDAG $\mathcal{G}_j$ for each $\lambda_j \in \Omega_{U}(\mathbf{A})$, and in the process enables each regular node $i$ to determine the set $\mathcal{N}^{(j)}_i$ for each $\lambda_j \in \overline{\mathcal{O}}_i$. The construction of these MEDAGs constitutes the initialization phase of our design, which can then be followed up by the LFRE algorithm described earlier. We briefly describe the implementation of Algorithm 1 as follows.

Algorithm 1 requires each node $i$ to maintain a counter $c_i(j)$ and a list of indices $\mathcal{N}^{(j)}_i$ for each $\lambda_j \in \Omega_{U}(\mathbf{A})$. The nodes in $\mathcal{N}^{(j)}_i \subseteq \mathcal{N}_i$ will be the parents of node $i$ in the DAG constructed for the estimation of $\mathbf{z}^{(j)}[k]$. Algorithm 1 is initialized with $c_i(j)=0$ and $\mathcal{N}^{(j)}_i = \emptyset$, for each $i\in\mathcal{V}$. Subsequently, the algorithm proceeds in rounds where in the first round each node in $\mathcal{S}_j$ broadcasts the message $``1"$ to its {out-neighbors}, sets $c_i(j)=1$, maintains $\mathcal{N}^{(j)}_i = \emptyset$ for all future rounds, and goes to sleep. Each node $i \in \mathcal{V}\setminus\mathcal{S}_j$ waits until it has received $``1"$ from at least $(2f+1)$ distinct neighbors, at which point it sets $c_i(j)=1$, appends the labels of each of the neighbors from which it received $``1"$ to $\mathcal{N}^{(j)}_i$, broadcasts the message $``1"$ to its {out-neighbors}, and goes to sleep. Let $\mathcal{R}'\subseteq\mathcal{V}$ denote the set of nodes that behave regularly during the execution of Algorithm 1. We say that the MEDAG construction algorithm ``\textit{terminates for $\lambda_j$}" if there exists $T_j \in \mathbb{N}_{+}$ such that $c_i(j) =1$ $\forall i \in \mathcal{R}'$,  for all rounds following round $T_j$. The \textbf{objective} of the algorithm is to return a set of sets $\{\mathcal{N}^{(j)}_i\}$, where $\lambda_j \in \Omega_{U}(\mathbf{A})$, and $i \in \mathcal{V}$.

\begin{algorithm}[t]
 \caption{MEDAG Construction Algorithm}
 \label{algo:DAG}
 \begin{algorithmic}
 \vspace{2mm}
 \STATE  For each eigenvalue $\lambda_j \in \Omega_{U}(\mathbf{A})$ \textbf{do}:
 \STATE \textbf{Initialization}: Initialize $c_i(j)=0, \, \mathcal{N}^{(j)}_i = \emptyset, \, \forall i \in \mathcal{{V}}$. Each node determines whether it belongs to $\mathcal{S}_j$.
 \STATE \textbf{Actions of the source nodes}: Each node in  $\mathcal{S}_j$ updates its counter value $c_i(j)=1$, and transmits the message $``1"$ to its {out-neighbors}. Following this step, it does not listen to any other node, i.e., $\mathcal{N}^{(j)}_i = \emptyset$ and $c_i(j)=1$, $\forall i \in \mathcal{S}_j$ for the remainder of the algorithm.
\STATE \textbf{Actions of the non-source nodes}: Each node $i \in \mathcal{V} \setminus \mathcal{S}_j$ does the following: 
 \begin{itemize}
 \item If $c_i(j) = 0$ \textit{and} node $i$  has received $``1"$ from at least $(2f+1)$ distinct neighbors (not necessarily all in the same round), it updates $c_i(j)$ to $1$, appends the labels of the neighbors from which it received $``1"$ to $\mathcal{N}^{(j)}_i$, and transmits $``1"$ to its {out-neighbors}. 
 \item If $c_i(j) = 1$, it discards all messages received from its neighbors, i.e., it does not update $c_i(j)$ or $\mathcal{N}^{(j)}_i$.
 \end{itemize}
\STATE \textbf{Return} : A set of sets $\{\mathcal{N}^{(j)}_i\}, \lambda_j \in \Omega_{U}(\mathbf{A})$, $i \in \mathcal{V}$.
 \end{algorithmic}
 \end{algorithm}
 
We emphasize that in addition to misbehavior during the state estimation phase (run-time), an adversarial node is allowed to misbehave during the implementation of Algorithm 1 (design-time) as well. For example, it can transmit the message out of turn, i.e., before receiving $``1"$ from at least $(2f+1)$ neighbors. It can also choose not to transmit the message at all. Note however that we must have $\mathcal{V}\setminus\mathcal{R}' \subseteq \mathcal{A}$, i.e., the $f$-local set of adversaries during the estimation phase must contain the set of adversaries during the design phase. In the next section, we shall detail graph conditions that guarantee the termination of the MEDAG construction algorithm under arbitrary adversarial behavior. For the following discussion, we characterize the properties of the output of Algorithm 1 if it terminates. To this end, consider the spanning subgraph $\mathcal{G}_j=(\mathcal{V},\mathcal{E}_j)$  induced by the sets $\{\mathcal{N}^{(j)}_i\}$ returned by Algorithm 1. Keeping in mind that $\mathcal{R}'\supseteq \mathcal{R}$ represents the set of nodes that behave regularly during the execution of Algorithm 1, we have the following results; proofs of these results can be found in \cite{mitraCDC}.
\begin{proposition}
Suppose Algorithm 1 terminates for some $\lambda_j \in \Omega_{U}(\mathbf{A})$, and returns the sets $\{\mathcal{N}^{(j)}_i\}$. Then, the spanning subgraph $\mathcal{G}_j$ induced by these sets contains no directed cycles where every node belongs to $\mathcal{R}'$.
\label{prop:MEDAG}
\end{proposition}

Let $\mathcal{L}^{(j)}_{m-1}$ denote the set of all nodes in $\mathcal{R}'$ that update their counter value from $0$ to $1$ in round $m$ of Algorithm 1, i.e., $\mathcal{L}^{(j)}_{0} = \mathcal{S}_j \cap \mathcal{R}'$, and so on.\footnote{Note that the method developed in this paper allows even some of the source nodes in $\mathcal{S}_j$ to be adversarial.}
\begin{proposition}
Suppose Algorithm 1 terminates for some $\lambda_j \in \Omega_{U}(\mathbf{A})$. Let $T_j$ denote the smallest integer such that in round $T_j$, $c_i(j)=1$ $\forall i \in \mathcal{R}'$. Then, the sets $\{\mathcal{L}^{(j)}_{0}, \cdots, \mathcal{L}^{(j)}_{T_j}\}$ form a partition of the set $\mathcal{R}'$ in $\mathcal{G}_j$. 
\label{prop:partition}
\end{proposition}
\begin{thm}
Suppose the MEDAG construction algorithm terminates for $\lambda_j \in \Omega_{U}(\mathbf{A})$. Then, there exists a subgraph $\mathcal{G}_j$ satisfying the properties of a MEDAG for all $f$-local sets $\mathcal{A}$ that contain $\mathcal{V}\setminus\mathcal{R}'$ as a subset.
\label{theo:MEDAG}
\end{thm}
\vspace{-0.5cm}
\begin{pf}
The result follows immediately from Propositions \ref{prop:MEDAG} and \ref{prop:partition}. \qed
\end{pf}
\begin{remark}
Based on the above theorem, we make the following observations. If Algorithm 1 terminates for each $\lambda_j \in \Omega_{U}(\mathbf{A})$, and $\mathcal{V}\setminus\mathcal{R}'=\emptyset$, then the $\mathcal{G}_j$ subgraphs satisfy all the properties of a MEDAG and we can directly invoke Theorem \ref{thm:consensus}. If Algorithm 1 terminates for each $\lambda_j \in \Omega_{U}(\mathbf{A})$ and $\mathcal{V}\setminus\mathcal{R}'\neq\emptyset$, (i.e., there is some adversarial activity during the MEDAG construction phase), then we do not need to provide any guarantees of state estimation for the set of misbehaving nodes $\mathcal{V}\setminus\mathcal{R}'$, since $\mathcal{V}\setminus\mathcal{R}' \subseteq\mathcal{A}$. In this case too, the subgraphs returned by Algorithm 1 have enough redundancy to ensure that Problem 1 can be solved based on our proposed approach; this fact can be established using arguments identical to those used for proving Theorem \ref{thm:consensus}. In what follows, we summarize our overall approach.
\label{rem:clarify}
\end{remark}
\subsection{Summary of the Resilient Distributed State Estimation Scheme}
\label{sec:summary}
\begin{itemize}
\item[1)] Each regular node $i\in\mathcal{R}$ performs the coordinate transformation $\mathbf{z}[k]= \mathbf{T}^{-1} \mathbf{x}[k]$ described in Section \ref{sec:prelim}; accordingly, it identifies its detectable and undetectable eigenvalues ($\mathcal{O}_{i}$ and {$\overline{\mathcal{O}}_i$} respectively). 
\item[2)] The MEDAG construction algorithm described by Algorithm 1 is implemented for each $\lambda_j\in\Omega_{U}(\mathbf{A})$; graph conditions for termination of this algorithm are provided in the next section. At the end of this algorithm, each regular node $i$ knows the subset $\mathcal{N}^{(j)}_i$ of neighbors it should use in the LFRE algorithm.
\item[3)] Each regular node $i$ employs a locally constructed Luenberger observer (refer to Lemma \ref{lemma:luen} and the discussion preceding it) for estimating $\mathbf{z}_{\mathcal{O}_i}[k]$, namely the portion of the state $\mathbf{z}[k]$ corresponding to its detectable eigenvalues.
\item[4)] Each regular node $i$ employs the LFRE algorithm governed by equations \eqref{eqn:update rule1} and \eqref{eq:update2} for estimating $\mathbf{z}_{\overline{\mathcal{O}}_i}[k]$, namely the portion of the state $\mathbf{z}[k]$ corresponding to its undetectable eigenvalues.
\end{itemize}

\begin{remark}
Whereas steps 1 and 2 correspond to the initial design phase of our scheme, steps 3 and 4 constitute the estimation phase. A key benefit of the proposed method is that if certain graph-theoretic conditions (to be discussed in the following section) are met, then our overall scheme provably admits a \textbf{fully distributed implementation} even under worst-case adversarial behavior.
\end{remark}

\section{Feasible Graph Topologies}
\label{sec:feasible}
In this section, we characterize  feasible graph topologies that guarantee the termination of the MEDAG construction algorithm described in the previous section. In other words, based on Remark \ref{rem:clarify}, feasible graph topologies guarantee that Problem 1 can be solved based on our proposed approach (summarized in Section \ref{sec:summary}). We first recall the following definition from \cite{WMSR,rescons}.
\begin{definition}($r$-\textbf{reachable set}) For a graph $\mathcal{G}=(\mathcal{V,E})$, a set $\mathcal{S} \subset \mathcal{V}$, and an integer $r \in \mathbb{N}_{+}$, $\mathcal{S}$ is an \textit{$r$-reachable set} if there exists an $i \in \mathcal{S}$ such that $|\mathcal{N}_i \setminus \mathcal{S}| \geq r$, 
\end{definition}

Thus, if a set $\mathcal{S}$ is $r$-reachable, then it contains a node which has at least $r$ neighbors outside $\mathcal{S}$. We modify the notion of a \textit{strongly-r robust graph} from \cite{WMSR} as follows.

\begin{definition}(\textbf{strongly} $r$-\textbf{robust graph} \textit{w.r.t.} $\mathcal{S}_j$) For $r \in \mathbb{N}_{+}$ and $\lambda_j \in \Omega_{U}(\mathbf{A})$, a graph $\mathcal{G}=(\mathcal{V,E})$ is \textit{strongly $r$-robust w.r.t. to the set of source nodes $\mathcal{S}_j$}, if for any non-empty subset $\mathcal{C} \subseteq \mathcal{V}\setminus\mathcal{S}_j$, $\mathcal{C}$ is $r$-reachable.
\label{defn:strongrobust}
\end{definition}

{For an illustration of the above definitions, the reader is referred back to Figure  \ref{fig:sim_DAG}. Figure \ref{fig:sim_DAG}(a) is an example of a network that is strongly $3$-robust w.r.t. the set of source nodes, namely nodes $\{1,2,3\}$. Specifically, all subsets of $\{4,5,6,7\}$ are $3$-reachable (i.e., each such subset has a node that has at least 3 neighbors outside that subset).}

\begin{lemma} 
\label{lemma:type1}
The MEDAG construction algorithm terminates for $\lambda_j \in \Omega_U(\mathbf{A})$ if $\mathcal{G}$ is strongly $(3f+1)$-robust w.r.t. $\mathcal{S}_j$.
\end{lemma}
\vspace{-3mm}
\begin{proof}
We prove by contradiction. Consider any $\lambda_j \in \Omega_U(\mathbf{A})$ and let $\mathcal{G}$ be strongly $(3f+1)$-robust w.r.t. the set of source nodes $\mathcal{S}_j$. Suppose that the MEDAG construction algorithm for $\lambda_j$ does not terminate. Since the possibility of the counter $c_i(j)$ oscillating between $0$ and $1$ (where $i \in \mathcal{R})$ is ruled out based on our MEDAG construction algorithm, there must then exist a non-empty set $\mathcal{C} \subseteq \mathcal{V}\setminus\mathcal{S}_j$ of regular nodes that never update their counter $c_i(j)$ from $0$ to $1$, where $i\in\mathcal{C}$. As $\mathcal{G}$ is strongly $(3f+1)$-robust w.r.t. $\mathcal{S}_j$, it follows that $\mathcal{C}$ is $(3f+1)$-reachable, i.e., there exists a node $i \in \mathcal{C}$ which has at least $(3f+1)$  neighbors outside $\mathcal{C}$. Under the $f$-local adversarial model, at least $(2f+1)$ of them are regular nodes with $c_i(j) = 1$. Thus, at least $(2f+1)$ regular nodes must have transmitted the message $``1"$ to node $i$. Thus, based on the rules of Algorithm 1, node $i$ must have updated $c_i(j)$ from $0$ to $1$ at some point of time, leading to a contradiction.
\end{proof}
\vspace{-3mm}
Whereas the $(2f+1)$ term appears in various contexts when dealing with security problems on networks (such as distributed consensus \cite{rescons,vaidyacons}, broadcasting \cite{flocal1,flocal2} and optimization \cite{su,Sundaramopt}), the $(3f+1)$ term featuring in our analysis accounts for misbehavior that involves transmission of no messages by the adversarial nodes during execution of the MEDAG construction algorithm described in Section \ref{sec:MEDAG}.
We now present the main result of this paper which ties together the previous results presented in this paper, and in turn provides a connection between feasible graph topologies and the solution to Problem 1 based on our proposed approach.

\begin{thm}
\label{theo:Feasiblegraph} Consider an LTI system \eqref{eqn:plant} and a measurement model \eqref{eqn:Obsmodel}. Let the communication graph $\mathcal{G}$ be strongly $(3f+1)$-robust w.r.t. $\mathcal{S}_j, \forall \lambda_j\in\Omega_{U}(\mathbf{A})$. Then, the proposed algorithm summarized in Section \ref{sec:summary} provides a solution to Problem 1.
\label{thm:main}
\end{thm}

\begin{proof}
From Lemma \ref{lemma:type1}, it follows that if $\mathcal{G}$ is strongly $(3f+1)$-robust w.r.t. $\mathcal{S}_j$ for every $\lambda_j \in \Omega_{U}(\mathbf{A})$, then the MEDAG construction algorithm terminates for each such eigenvalue. Combining Theorem \ref{theo:MEDAG}, Remark \ref{rem:clarify} and Theorem \ref{thm:consensus} then leads to the desired result.
\end{proof}
If the adversarial attacks are restricted to the estimation phase only, i.e., if there are no attacks during the initial MEDAG construction phase, then the following result provides a tight graph condition for our algorithm.

\begin{thm} Consider an LTI system \eqref{eqn:plant} and a measurement model \eqref{eqn:Obsmodel}. Suppose adversarial behavior is restricted to the estimation phase (steps $3$ and $4$) of the proposed algorithm summarized in Section \ref{sec:summary}. Then, this algorithm solves Problem $1$  if and only if $\mathcal{G}$ is strongly $(2f+1)$-robust w.r.t. $\mathcal{S}_j$, $\forall \lambda_j \in \Omega_{U}(\mathbf{A})$.
\label{thm:tight}
\end{thm}

The proof of the above result is given in Appendix \ref{app:theotight}. Essentially, Theorem \ref{thm:tight} alludes to the fact that $\mathcal{G}$ contains a MEDAG $\mathcal{G}_j$ for each $\lambda_j \in \Omega_{U}(\mathbf{A})$ \textit{if and only if} $\mathcal{G}$ is strongly $(2f+1)$-robust w.r.t. $\mathcal{S}_j$, $\forall \lambda_j \in \Omega_{U}(\mathbf{A})$. Note that although Theorem \ref{thm:tight} provides a graph condition that is necessary and sufficient for the algorithm developed in this paper, such a condition may not be necessary for solving Problem 1 in general.

Theorems \ref{thm:main} and \ref{thm:tight} reveal that `\textit{strong $r$-robustness w.r.t. $\mathcal{S}_j,\ \forall \lambda_j\in\Omega_{U}(\mathbf{A})$}' is the key topological property required for guaranteeing success of our proposed algorithm. Accordingly, given a system model \eqref{eqn:plant} and measurement model \eqref{eqn:Obsmodel}, a network that is strongly $r$-robust w.r.t. $\mathcal{S}_j, \ \forall \lambda_j\in\Omega_{U}(\mathbf{A})$, will be called an `\textit{$r$-feasible network}' for simplicity. We summarize certain features of an $r$-feasible network in the following result.

\begin{proposition} An $r$-feasible network $\mathcal{G}$ has the following properties.
\begin{itemize}
\item[(i)] The graph $\mathcal{G}^{'}=(\mathcal{V}\cup{v}_{new},\mathcal{E}\cup\mathcal{E}_{new})$, where $v_{new}$ is a new vertex added to $\mathcal{G}$ and $\mathcal{E}_{new}$ is the edge set associated with ${v}_{new}$, is an $r$-feasible network if $|\mathcal{N}_{v_{new}}| \geq r.$
\item[(ii)] $r \leq \min_{\lambda_j \in \Omega_{U}(\mathbf{A})}|\mathcal{S}_j|$.
\item[(iii)] Let $\mathcal{S}=\bigcap_{\lambda_j \in \Omega_{U}(\mathbf{A})}\mathcal{S}_j$. Then $|\mathcal{N}_i| \geq r, \, \forall i \in \mathcal{V}\setminus\mathcal{S}.$
\item[(iv)] Removal of a $k$-local set from $\mathcal{G}$, where $0 < k < r$, results in a network $\mathcal{G}^{'}$ that is $(r-k)$-feasible.
\end{itemize}
\label{prop:properties}
\end{proposition}
\begin{proof} (i) Consider any $\lambda_j \in \Omega_{U}(\mathbf{A})$. If $v_{new}$ is a source node for $\lambda_j$, then it is easily seen that $\mathcal{G}^{'}$ is strongly $r$-robust w.r.t. $\mathcal{S}_j$. For the case when $v_{new}$ is not a source node for $\lambda_j$, consider any non-empty set $\mathcal{C}\subseteq \{\mathcal{V}\cup{v_{new}}\}\setminus\mathcal{S}_j$. If $\mathcal{C}=\{v_{new}\}$, then $r$-reachability of $\mathcal{C}$ follows from the fact that $|\mathcal{N}_{v_{new}}| \geq r$. In every other case, $\mathcal{C}$ contains some nodes of the original graph $\mathcal{G}$ and is hence $r$-reachable as $\mathcal{G}$ is $r$-feasible. Thus $\mathcal{G}^{'}$ is strongly $r$-robust w.r.t. $\mathcal{S}_j$. A similar analysis holds for each $\lambda_j \in {\Omega}_{U}(\mathbf{A})$, leading to the desired result.

(ii) Suppose $r > |\mathcal{S}_{\rho}|$  where $\rho=\arg \min_{\lambda_j \in \Omega_{U}(\mathbf{A})}|\mathcal{S}_j|$. Since the set $\mathcal{C}=\mathcal{V}\setminus\mathcal{S}_{\rho}$ can be at most $|\mathcal{S}_{\rho}|$-reachable and $|\mathcal{S}_{\rho}| < r$, it follows that $\mathcal{G}$ is not $r$-feasible.

(iii) Suppose $i \in \mathcal{V}\setminus\mathcal{S}$ with $|\mathcal{N}_i| < r$. As $i \in \mathcal{V}\setminus\mathcal{S}$, there exists some $\lambda_j \in \Omega_{U}(\mathbf{A})$ such that $i\in \mathcal{V}\setminus\mathcal{S}_j$. Consider the set $\mathcal{C}=\{i\}$. As $|\mathcal{N}_i| < r$, the set $\mathcal{C}$ is not $r$-reachable. Thus, $\mathcal{G}$ is not strongly $r$-robust w.r.t. $\mathcal{S}_j$, implying that $\mathcal{G}$ is not $r$-feasible.

(iv) First, observe that as $\mathcal{G}$ is $r$-feasible and $k < r$, removal of a $k$-local set from $\mathcal{G}$ cannot cause the removal of an entire source node net $\mathcal{S}_j$ for any $\lambda_j\in\Omega_{U}(\mathbf{A})$. This follows from noting that any source set $\mathcal{S}_j$ where $\lambda_j \in \Omega_{U}(\mathbf{A})$ (or any set containing $\mathcal{S}_j$) will have an overlap of at least $r$ nodes with the neighborhood of some non-source node owing to the $r$-feasibility of the original network. As $r>k$, such sets are not $k$-local. Next, pick any $\lambda_j\in\Omega_{U}(\mathbf{A})$ and let $\mathcal{C}$ be a non-empty subset of $\mathcal{V}^{'}\setminus{\mathcal{S}_j}^{'}$, where $\mathcal{V}^{'}$ and ${\mathcal{S}_j}^{'}$ represent the vertex set and source node set for $\lambda_j$, respectively, in $\mathcal{G}^{'}$. Since $\mathcal{C}$ was $r$-reachable in $\mathcal{G}$, it contained some node $v$ with $r$ neighbors outside $\mathcal{C}$. While constructing $\mathcal{G}^{'}$, node $v$ can lose at most $k$ of such neighbors, and hence $\mathcal{C}$ is $(r-k)$-reachable in $\mathcal{G}^{'}$. The rest of the proof follows trivially.
\end{proof}
We remark on certain implications of the above result. The first property provides a procedure for constructing $r$-feasible networks with $N$ nodes (where $N > r$) starting from $r$-feasible networks with fewer than $N$ nodes. The second property shows that the measurement structure of the nodes provides an upper bound on the robustness of the overall network. The third property places a lower bound on the minimum in-degree of any node that cannot estimate the entire state on its own in an $r$-feasible network. Finally, a direct implication of the fourth property is that a loss of $k$ source nodes (where $k < r$) for any unstable eigenvalue of the system (possibly due to sensor failures) leaves the resulting network at least $(r-k)$-feasible if the original network is $r$-feasible to begin with.
\textbf{Applicability of the Proposed Approach}:
Building on the insights developed in this section, we make a case for the applicability of the approach developed in this paper by addressing the following question: How efficiently can one verify whether a given system and  network is $r$-feasible? To answer the above question we will exploit a connection between the `strong $r$-robustness property w.r.t. a certain set of nodes' and the dynamic process of `bootstrap percolation' on networks \cite{boot1}. Given a graph $\mathcal{G}$ and a threshold $r \geq 2$, bootstrap percolation can be viewed as a process of spread of \textit{activation} where one starts off with a set $\mathcal{I}\subseteq\mathcal{V}$ of initially active nodes. Subsequently, the process evolves over the network based on the rule that an inactive node becomes active if and only if it has at least $r$ active neighbors, with active nodes remaining active forever. The process terminates when no more nodes become active; an initial set $\mathcal{I}$ is said to \textit{percolate} if upon termination the final active set equals the entire node set $\mathcal{V}$. Consider the following simple, yet key observation.

\begin{lemma} Given a graph $\mathcal{G}$ and a threshold $r \geq 2$, an initial set $\mathcal{I}$ percolates via the process of bootstrap percolation if and only if $\mathcal{G}$ is strongly $r$-robust w.r.t. $\mathcal{I}$.
\label{lemma:bootstrap}
\end{lemma} 
The proof of the above result follows similar arguments as Lemma \ref{lemma:type1}, and is hence omitted. Leveraging Lemma \ref{lemma:bootstrap}, we obtain the following result.
\begin{proposition}
Given a system matrix $\mathbf{A} \in \mathbb{R}^{n \times n}$ \eqref{eqn:plant}, a measurement model \eqref{eqn:Obsmodel}, a communication graph $\mathcal{G}=(\mathcal{V,E})$ with $|\mathcal{V}|=N$, the source set $\mathcal{S}_j$ for each $\lambda_j\in sp(\mathbf{A})$, and an integer $r\geq 2$, one can verify whether the network is $r$-feasible in O$(nN|\mathcal{E}|)$ time.
\end{proposition}
\vspace{-3mm}
\begin{proof}
Notice that $|\Omega_{U}(\mathbf{A})| \leq n$, i.e., there are at most $n$ source sets $\mathcal{S}_j$ for which we need to verify the strong $r$-robustness property in Definition \ref{defn:strongrobust}. Based on Lemma \ref{lemma:bootstrap}, for each $\mathcal{S}_j$ corresponding to some $\lambda_j \in \Omega_{U}(\mathbf{A})$, verifying whether $\mathcal{G}$ is strongly $r$-robust w.r.t. $\mathcal{S}_j$ is equivalent to verifying whether $\mathcal{S}_j$ percolates via the process of bootstrap percolation with threshold $r$. Thus, we analyze the complexity of simulating a bootstrap percolation process on a given network.\footnote{ Algorithm 1 essentially simulates the evolution of a bootstrap percolation process with threshold $r=(2f+1)$, provided there is no adversarial activity during the distributed implementation of such an algorithm.} First, notice that it takes at most $N$ iterations/rounds for a bootstrap percolation process to terminate on a network of $N$ nodes. In each round, every inactive node checks whether it has at least $r$ active neighbors; the entire process of checking is thus completed in \textit{O}$(\sum^{N}_{i=1}d_i)$ = \textit{O}$(|\mathcal{E}|)$ time, where $d_i$ represents the in-degree of node $i$. Thus, for a given initial set, it takes \textit{O}$(N|\mathcal{E}|)$ time to simulate the bootstrap percolation process. The result then follows readily.\end{proof}
\begin{remark} Based on the above result, one can check whether the approach developed in this paper is applicable for a given system and network in polynomial time.\footnote{This result is in stark contrast with analogous results existing in the resilient distributed consensus \cite{rescons,vaidyacons} and optimization \cite{Sundaramopt,su} literature, since checking the `robustness' condition needed for solving such problems is coNP-complete.}Interestingly, leveraging the equivalence described in Lemma \ref{lemma:bootstrap}, it is possible to show that the strong $r$-robustness property described in Definition \ref{defn:strongrobust} is exhibited by various large-scale complex network models such as 
 the Barab\' asi-Albert (BA) preferential attachment model, the Erd\H os-R\'enyi random graph model, and the $2$-dimensional random geometric graph model. A detailed discussion on this topic can be found in Appendix \ref{app:random}.
\end{remark}
\section{Simulations}
\begin{figure}[t]
\hspace*{-6mm} 
\begin{tabular}{c c}
\includegraphics[scale=0.09]{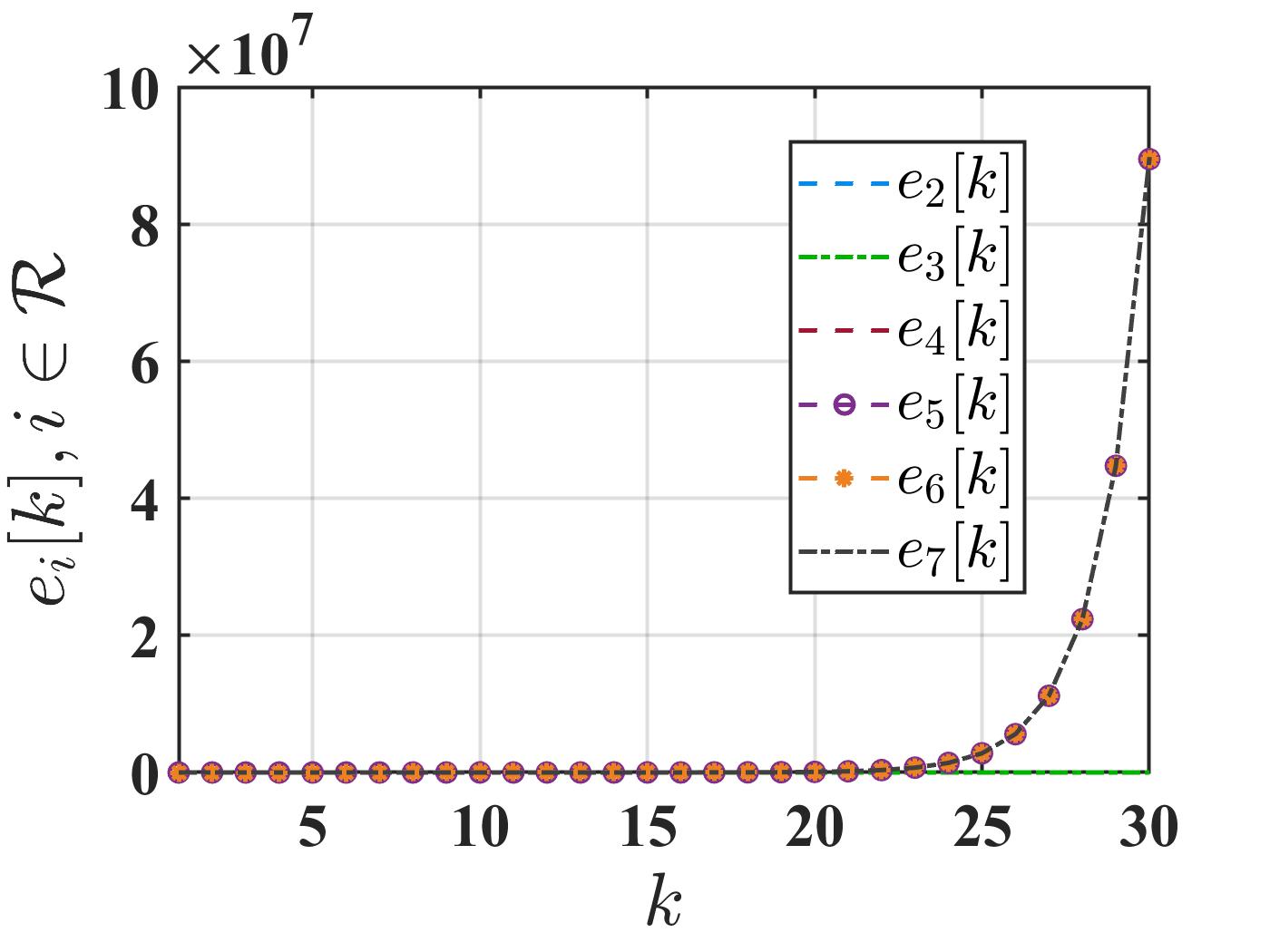}&\hspace{-3mm}\includegraphics[scale=0.09]{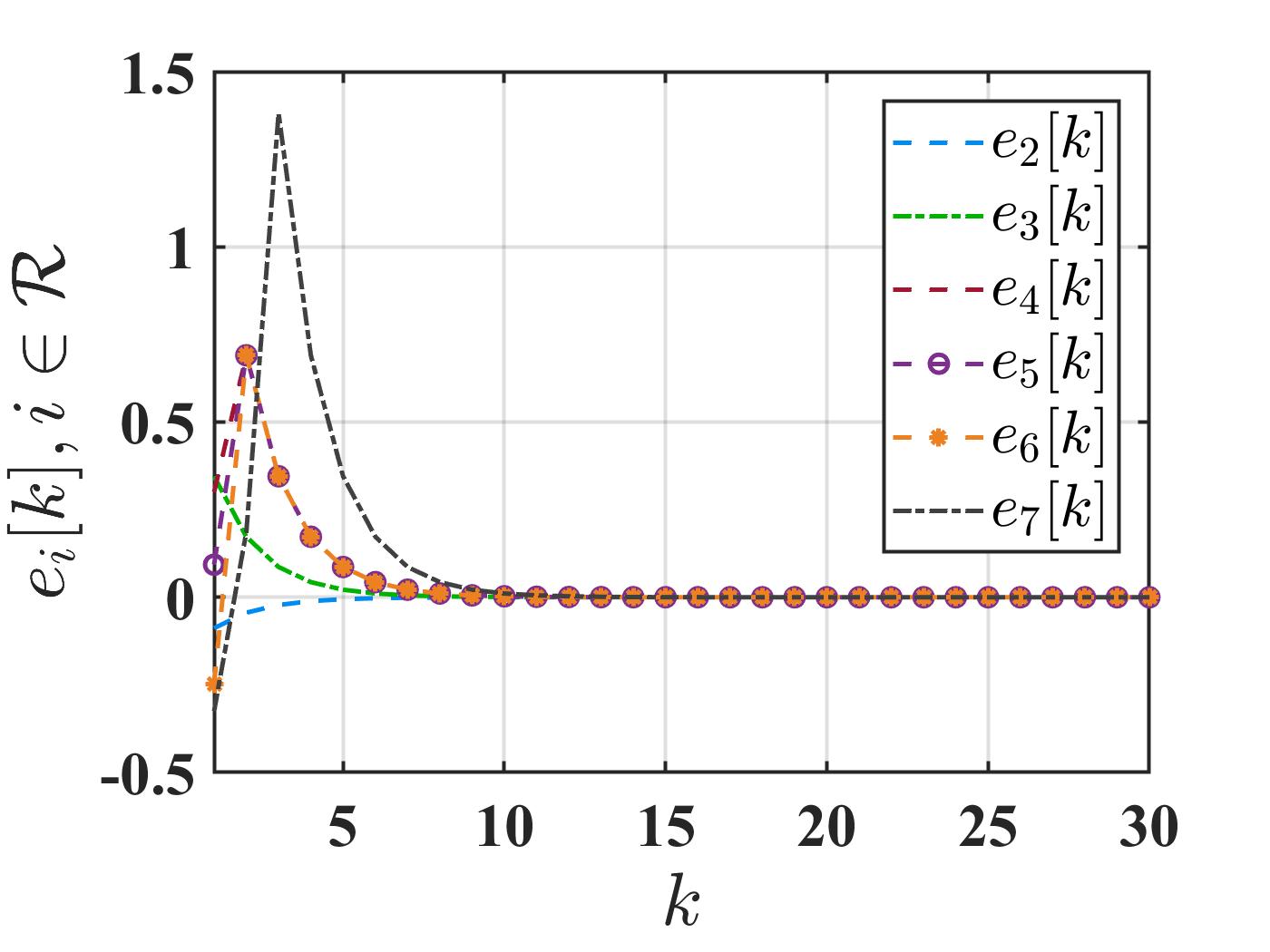}\\
(a)&(b)
\end{tabular}
\caption{{Consider the system and network in Fig. \ref{fig:sim_DAG}. Fig. (a) depicts how a single adversary, namely node 1, can cause the estimation errors of all the non-source regular nodes (namely, nodes 4-7) to diverge when a non-resilient distributed observer is employed. Fig. (b) shows hows the proposed LFRE algorithm counteracts the effect of the adversary.}}
\label{fig:sim}
\end{figure}
\vspace{-3mm}
{
Consider the system and network given by Figure \ref{fig:sim_DAG}. The state evolves as $x[k+1]=ax[k]$, with $a=2$. Nodes 1, 2 and 3 are the source nodes and directly estimate the state, i.e., $y_i[k]=x[k], \forall i\in \{1,2,3\}$. The rest of the nodes have zero measurements. Node $1$ is the only adversarial node in the network, and it simply transmits a constant signal of magnitude $\epsilon=0.001$ to each of its neighbors at every time-step. Each regular source node updates its state estimate based on a standard Luenberger observer as follows:
\begin{equation}
\hat{x}_i[k+1]=a\hat{x}_i[k]+l_i(y_i[k]-\hat{x}_i[k]), i\in\{2,3\},
\label{eqn:sim1}
\end{equation}
with the observer gain $l_i$ set to $1.5$ (this gain is simply chosen to ensure stability). We first consider a scenario where each non-source node updates its estimate as follows: $\hat{x}_i[k+1]=a\sum_{l\in\mathcal{N}^{(j)}_i}w^{(j)}_{il}\hat{x}_l[k]$, where the weights form a convex combination, and $\mathcal{N}^{(j)}_i$ represents the neighbors of node $i$ in the MEDAG shown in Figure \ref{fig:sim_DAG}(b).\footnote{For this scalar system, there is only one mode, i.e., $j=1$.} If node $1$ were to update its estimate as per \eqref{eqn:sim1}, then it can be easily verified analytically that all nodes would be able to track the state asymptotically (see \cite{mitraarchive} for details). However, as seen from Figure \ref{fig:sim}(a), a single adversarial node (node 1 in this case) transmitting a small constant signal can cause the estimates of all the non-source nodes to diverge. This example demonstrates that although the underlying network is strongly $3$-robust (i.e., has enough built-in redundancy to deal with a single adversarial node\footnote{For this example, we assume that node 1 misbehaves only during the estimation phase. Hence, Theorem \ref{thm:tight} is applicable.}), the non-resilient distributed observer employed above proves to be inadequate in the face of attacks. However, as seen from Figure \ref{fig:sim}(b), the LFRE algorithm  complements the robust network structure and succeeds in counteracting the adversarial attack. For all simulations, $x[0]=0.5$, and $\hat{x}_i[0], i\in\mathcal{V}$ is a random number between $0$ and $1$.}
\section{Conclusions}
\vspace{-4mm}
We studied the problem of collaboratively estimating the state of an LTI system subject to worst-case adversarial behavior. For the attack models under consideration, we identified certain necessary conditions that need to be satisfied by any system and network for the problem posed in this paper to have a feasible solution. We then developed a local-filtering algorithm to enable each non-compromised node to recover the entire state. Finally, using a topological property called strong $r$-robustness, we characterized networks that guarantee success of our proposed strategy. Two notable features of our approach are as follows: (i) each step of our approach admits an attack-resilient, completely distributed implementation provided certain graph-theoretic conditions are met; and (ii) these graph-theoretic conditions can be checked in polynomial time as discussed in the previous section.

There are various interesting directions for future research, some of which were pointed out in Remark \ref{rem:directions}. Finding an algorithm-independent necessary and sufficient condition for the problem posed in this paper will likely be a challenging proposition. Whereas the focus of this paper has been on obtaining sufficient graph-theoretic conditions that account for worst-case adversarial behavior, it might be of interest to see if such conditions can be relaxed when confronted with less sophisticated adversarial attacks. Extensions of our framework to account for network-induced issues such as packet-drops, delays and asynchronicity also merit attention; see \cite{mitraACC18} for preliminary results on this topic.
\bibliographystyle{unsrt}
\bibliography{refs} 

\begin{thebibliography}{10}

\bibitem{survey1}
Deborah Estrin, Ramesh Govindan, John Heidemann, and Satish Kumar.
\newblock Next century challenges: Scalable coordination in sensor networks.
\newblock In {\em Proceedings of the 5th annual ACM/IEEE international
  conference on Mobile computing and networking}, pages 263--270. ACM, 1999.

\bibitem{martins3}
Shinkyu Park and Nuno~C Martins.
\newblock Design of distributed {LTI} observers for state omniscience.
\newblock {\em IEEE Transactions on Automatic Control}, 62(2):561--576, 2017.

\bibitem{allerton}
Aritra Mitra and Shreyas Sundaram.
\newblock An approach for distributed state estimation of {LTI} systems.
\newblock In {\em {P}roceedings of the 2016 54th Annual Allerton Conference on
  Communication, Control, and Computing}, pages 1088--1093, 2016.

\bibitem{mitraarchive}
A.~Mitra and S.~Sundaram.
\newblock Distributed observers for {LTI} systems.
\newblock {\em IEEE Transactions on Automatic Control}, 2018.

\bibitem{wang}
L~Wang and AS~Morse.
\newblock A distributed observer for a time-invariant linear system.
\newblock {\em IEEE Transactions on Automatic Control}, 63(7):2123--2130, 2018.

\bibitem{han}
Weixin Han, Harry~L Trentelman, Zhenhua Wang, and Yi~Shen.
\newblock A simple approach to distributed observer design for linear systems.
\newblock {\em IEEE Transactions on Automatic Control}, 2018.

\bibitem{kim}
Taekyoo Kim, Hyungbo Shim, and Dongil~Dan Cho.
\newblock {Distributed {L}uenberger observer design}.
\newblock In {\em {P}roceedings of the 55th IEEE Conference on Decision and
  Control}, pages 6928--6933, 2016.

\bibitem{Millan}
AR~del Nozal, L~Orihuela, and P~Mill{\'a}an.
\newblock {Distributed consensus-based Kalman filtering considering subspace
  decomposition}.
\newblock {\em IFAC-PapersOnLine}, 50(1):2494--2499, 2017.

\bibitem{ren}
Shaocheng Wang and Wei Ren.
\newblock {On the convergence conditions of distributed dynamic state
  estimation using sensor networks: A unified framework}.
\newblock {\em IEEE Transactions on Control Systems Technology}, 2017.

\bibitem{infinite1}
Usman Khan and Jos{\'e}~MF Moura.
\newblock Distributing the {K}alman filter for large-scale systems.
\newblock {\em IEEE Transactions on Signal Processing}, 56(10):4919--4935,
  2008.

\bibitem{olfati1}
Reza Olfati-Saber.
\newblock Distributed {K}alman filter with embedded consensus filters.
\newblock In {\em {P}roceedings of the 44th IEEE Conference on Decision and
  Control and European Control Conference}, pages 8179--8184, 2005.

\bibitem{Baras}
Ion Matei and John~S Baras.
\newblock Consensus-based linear distributed filtering.
\newblock {\em Automatica}, 48(8):1776--1782, 2012.

\bibitem{batti1}
Giorgio Battistelli and Luigi Chisci.
\newblock {Kullback--Leibler average, consensus on probability densities, and
  distributed state estimation with guaranteed stability}.
\newblock {\em Automatica}, 50(3):707--718, 2014.

\bibitem{batti2}
Giorgio Battistelli, Luigi Chisci, Giovanni Mugnai, Alfonso Farina, and Antonio
  Graziano.
\newblock Consensus-based linear and nonlinear filtering.

\bibitem{kamal}
Ahmed~Tashrif Kamal, Jay~A Farrell, Amit~K Roy-Chowdhury, et~al.
\newblock {Information Weighted Consensus Filters and Their Application in
  Distributed Camera Networks.}
\newblock {\em IEEE Trans. Automat. Contr.}, 58(12):3112--3125, 2013.

\bibitem{fabio}
Fabio Pasqualetti, Florian D{\"o}rfler, and Francesco Bullo.
\newblock Attack detection and identification in cyber-physical systems.
\newblock {\em IEEE Transactions on Automatic Control}, 58(11):2715--2729,
  2013.

\bibitem{sundaramtac}
Shreyas Sundaram and Christoforos~N Hadjicostis.
\newblock Distributed function calculation via linear iterative strategies in
  the presence of malicious agents.
\newblock {\em IEEE Transactions on Automatic Control}, 56(7):1495--1508, 2011.

\bibitem{bai1}
Cheng-Zong Bai, Fabio Pasqualetti, and Vijay Gupta.
\newblock Data-injection attacks in stochastic control systems: Detectability
  and performance tradeoffs.
\newblock {\em Automatica}, 82:251--260, 2017.

\bibitem{teixeira}
Andr{\'e} Teixeira, Iman Shames, Henrik Sandberg, and Karl~Henrik Johansson.
\newblock A secure control framework for resource-limited adversaries.
\newblock {\em Automatica}, 51:135--148, 2015.

\bibitem{pajic}
Miroslav Pajic, Insup Lee, and George~J Pappas.
\newblock Attack-resilient state estimation for noisy dynamical systems.
\newblock {\em IEEE Transactions on Control of Network Systems}, 4(1):82--92,
  2017.

\bibitem{mishra}
Shaunak Mishra, Yasser Shoukry, Nikhil Karamchandani, Suhas~N Diggavi, and
  Paulo Tabuada.
\newblock Secure state estimation against sensor attacks in the presence of
  noise.
\newblock {\em IEEE Transactions on Control of Network Systems}, 4(1):49--59,
  2017.

\bibitem{mo}
Yilin Mo and Bruno Sinopoli.
\newblock {Secure estimation in the presence of integrity attacks}.
\newblock {\em IEEE Transactions on Automatic Control}, 60(4):1145--1151, 2015.

\bibitem{yuan}
Yuan Chen, Soummya Kar, and Jose~MF Moura.
\newblock {Resilient Distributed Estimation Through Adversary Detection}.
\newblock {\em IEEE Transactions on Signal Processing}, 2018.

\bibitem{varshney}
Wael Hashlamoun, Swastik Brahma, and Pramod~K Varshney.
\newblock {Mitigation of {B}yzantine Attacks on Distributed Detection Systems
  Using Audit Bits}.
\newblock {\em IEEE Transactions on Signal and Information Processing over
  Networks}, 4(1):18--32, 2018.

\bibitem{forti}
Nicola Forti, Giorgio Battistelli, Luigi Chisci, Suqi Li, Bailu Wang, and Bruno
  Sinopoli.
\newblock {Distributed joint attack detection and secure state estimation}.
\newblock {\em IEEE Transactions on Signal and Information Processing over
  Networks}, 4(1):96--110, 2018.

\bibitem{sec1}
Ion Matei, John~S Baras, and Vijay Srinivasan.
\newblock Trust-based multi-agent filtering for increased smart grid security.
\newblock In {\em {P}roceedings of the Mediterranean Conference on Control \&
  Automation}, pages 716--721, 2012.

\bibitem{sec3}
Usman Khan and Aleksandar~M Stankovic.
\newblock Secure distributed estimation in cyber-physical systems.
\newblock In {\em Proceedings of the IEEE International Conference on
  Acoustics, Speech and Signal Processing}, pages 5209--5213, 2013.

\bibitem{deghat}
Mohammad Deghat, Valery Ugrinovskii, Iman Shames, and C{\'e}dric Langbort.
\newblock Detection of biasing attacks on distributed estimation networks.
\newblock In {\em {P}roceedings of the IEEE Conference on Decision and
  Control}, pages 2134--2139, 2016.

\bibitem{fawzi}
Hamza Fawzi, Paulo Tabuada, and Suhas Diggavi.
\newblock Secure estimation and control for cyber-physical systems under
  adversarial attacks.
\newblock {\em IEEE Transactions on Automatic Control}, 59(6):1454--1467, 2014.

\bibitem{joao2}
Michelle~S Chong, Masashi Wakaiki, and Joao~P Hespanha.
\newblock Observability of linear systems under adversarial attacks.
\newblock In {\em {P}roceedings of the American Control Conference}, pages
  2439--2444, 2015.

\bibitem{mitraCDC}
Aritra Mitra and Shreyas Sundaram.
\newblock Secure distributed observers for a class of linear time invariant
  systems in the presence of {B}yzantine adversaries.
\newblock In {\em {P}roceedings of the 55th IEEE Conference on Decision and
  Control}, pages 2709--2714, 2016.

\bibitem{Byz}
Danny Dolev, Nancy~A Lynch, Shlomit~S Pinter, Eugene~W Stark, and William~E
  Weihl.
\newblock Reaching approximate agreement in the presence of faults.
\newblock {\em Journal of the ACM (JACM)}, 33(3):499--516, 1986.

\bibitem{vaidyacons}
Nitin~H Vaidya, Lewis Tseng, and Guanfeng Liang.
\newblock Iterative approximate {B}yzantine consensus in arbitrary directed
  graphs.
\newblock In {\em Proceedings of the ACM symposium on Principles of distributed
  computing}, pages 365--374, 2012.

\bibitem{rescons}
Heath~J LeBlanc, Haotian Zhang, Xenofon Koutsoukos, and Shreyas Sundaram.
\newblock Resilient asymptotic consensus in robust networks.
\newblock {\em IEEE Journal on Selected Areas in Communications},
  31(4):766--781, 2013.

\bibitem{Sundaramopt}
Shreyas Sundaram and Bahman Gharesifard.
\newblock Distributed optimization under adversarial nodes.
\newblock {\em IEEE Transactions on Automatic Control}, 2018.

\bibitem{su}
Lili Su and Nitin~H Vaidya.
\newblock Fault-tolerant multi-agent optimization: optimal iterative
  distributed algorithms.
\newblock In {\em Proceedings of the 2016 ACM Symposium on Principles of
  Distributed Computing}, pages 425--434. ACM, 2016.

\bibitem{flocal1}
Chiu-Yuen Koo.
\newblock Broadcast in radio networks tolerating {B}yzantine adversarial
  behavior.
\newblock In {\em Proceedings of the twenty-third annual ACM symposium on
  Principles of distributed computing}, pages 275--282. ACM, 2004.

\bibitem{flocal2}
Andrzej Pelc and David Peleg.
\newblock Broadcasting with locally bounded {B}yzantine faults.
\newblock {\em Information Processing Letters}, 93(3):109--115, 2005.

\bibitem{fabio2}
Fabio Pasqualetti, Antonio Bicchi, and Francesco Bullo.
\newblock Consensus computation in unreliable networks: A system theoretic
  approach.
\newblock {\em IEEE Transactions on Automatic Control}, 57(1):90--104, 2012.

\bibitem{Kosut}
Oliver Kosut, Liyan Jia, Robert~J Thomas, and Lang Tong.
\newblock Malicious data attacks on the smart grid.
\newblock {\em IEEE Transactions on Smart Grid}, 2(4):645--658, 2011.

\bibitem{hespanha}
Joao~P Hespanha.
\newblock {\em Linear systems theory}.
\newblock Princeton university press, 2009.

\bibitem{horn}
Roger~A Horn and Charles~R Johnson.
\newblock {\em Matrix analysis}.
\newblock Cambridge university press, 2012.

\bibitem{WMSR}
Haotian Zhang and Shreyas Sundaram.
\newblock Robustness of information diffusion algorithms to locally bounded
  adversaries.
\newblock In {\em {P}roceedings of the American Control Conference}, pages
  5855--5861, 2012.

\bibitem{boot1}
Svante Janson, Tomasz {\L}uczak, Tatyana Turova, Thomas Vallier, et~al.
\newblock Bootstrap percolation on the random graph {${G}_{N,P}$}.
\newblock {\em The Annals of Applied Probability}, 22(5):1989--2047, 2012.

\bibitem{mitraACC18}
Aritra Mitra and Shreyas Sundaram.
\newblock Secure distributed state estimation of an {LTI} system over
  time-varying networks and analog erasure channels.
\newblock In {\em {P}roceedings of the American Control Conference}, pages 6578
  -- 6583, 2018.

\bibitem{albert}
R{\'e}ka Albert and Albert-L{\'a}szl{\'o} Barab{\'a}si.
\newblock Statistical mechanics of complex networks.
\newblock {\em Reviews of modern physics}, 74(1):47, 2002.

\bibitem{erdos}
Paul Erd{\H{o}}s and Alfr{\'e}d R{\'e}nyi.
\newblock On the strength of connectedness of a random graph.
\newblock {\em Acta Mathematica Academiae Scientiarum Hungarica},
  12(1-2):261--267, 1964.

\bibitem{boot2}
Milan Bradonji{\'c} and Iraj Saniee.
\newblock Bootstrap percolation on random geometric graphs.
\newblock {\em Probability in the Engineering and Informational Sciences},
  28(2):169--181, 2014.

\bibitem{pottie}
Gregory~J Pottie and William~J Kaiser.
\newblock Wireless integrated network sensors.
\newblock {\em Communications of the ACM}, 43(5):51--58, 2000.

\end{thebibliography}
\appendix
\section{Proof of Theorem \ref{thm:ftotal}}
\label{app:proofthmftotal}
\begin{proof}
``(i)$\Longrightarrow$(ii)" We prove by contraposition. Suppose statement (ii) is violated for some node $i \in \mathcal{V}$, i.e., there exists a set $\mathcal{D}_i$ such that its removal from $\mathcal{G}$ causes the pair $(\mathbf{A},\mathbf{C}_{i\cup\mathcal{P}_i})$ to become undetectable (where $\mathcal{D}_i$ and $\mathcal{P}_i$ have the same meaning as in the statement of Theorem \ref{thm:ftotal}). It then follows that $\mathcal{F}=\mathcal{V}\setminus\{i\cup\mathcal{P}_i\}$ is a critical set. Suppose it is also a minimal critical set. We construct $\mathcal{G}^{'}$ by adding directed edges from a virtual node $s$ to each node in $\mathcal{F}$.\footnote{Throughout this proof, we drop the subscript on $\mathcal{G}^{'}, \mathcal{F}$ and $s$, unlike the notation in Section \ref{sec:fundamental}. This is done since the subscript $i$ is used to denote sets defined w.r.t. a node $i$ in this proof.} Observe that $\mathcal{H}=\mathcal{D}_i$  satisfies all the properties of an $f$-total pair cut w.r.t. $s$. In particular, $\mathcal{Y}=\{i\cup\mathcal{P}_i\}$ and $\mathcal{X}=\{\mathcal{V}\setminus{\{\mathcal{D}_i\cup\mathcal{Y}\}}\}\cup\{s\}.$ Thus, statement (i) is violated. A similar argument holds when $\mathcal{F}$ contains a minimal critical set.

``(i)$\Longleftarrow$(ii)" We again prove by contraposition. Suppose statement (i) is violated, i.e., there exists an $f$-total pair cut $\mathcal{H}$ w.r.t. a virtual node $s$ corresponding to some minimal critical set $\mathcal{F}$. Consider a node $i$ in $\mathcal{Y}$ (recall that $\mathcal{Y}$ is non-empty based on Definition \ref{defnt:cut}). First consider the case when node $i$ is not reachable from any node in $\mathcal{H}$ in the graph $\mathcal{G}$. It then follows that in the graph $\mathcal{G}$, directed paths to node $i$ can only exist from the set $\mathcal{Y}$. But since $i\in\mathcal{Y}$ and $(\mathbf{A},\mathbf{C}_{\mathcal{Y}})$ is not detectable, it is trivially impossible for node $i$ to estimate the state. We thus focus on the case where node $i$ is reachable from a certain set of nodes, say $\mathcal{D}_i$, within the set $\mathcal{H}$. Since $|\mathcal{H}|\leq 2f$ and $\mathcal{D}_i\subseteq\mathcal{H}$, we have that $|\mathcal{D}_i|\leq 2f$. It can be easily argued that the removal of $\mathcal{D}_i$ from $\mathcal{G}$ results in an induced subgraph where node $i$ can only be reached from the set $\mathcal{Y}$. In other words, the set $\mathcal{P}_i$, as defined in the statement of Theorem \ref{thm:ftotal}, is a subset of $\mathcal{Y}$. As $(\mathbf{A},\mathbf{C}_{\mathcal{Y}})$ is not detectable, it follows that $(\mathbf{A},\mathbf{C}_{i\cup\mathcal{P}_i})$ is not detectable either, and thus statement (ii) is violated. 

\end{proof}
\section{Proof of Theorem \ref{thm:consensus}}
\label{app:proofthm4}
\begin{proof} Let $\mathcal{A}$ be the (unknown) set of $f$-local adversaries, and consider $\mathcal{R}=\mathcal{V}\setminus\mathcal{A}$. Given a node $i\in\mathcal{R}$, the state vector $\mathbf{z}[k]$ can be partitioned into the components $\mathbf{z}_{\mathcal{O}_i}[k]$ and $\mathbf{z}_{\overline{\mathcal{O}}_i}[k]$ that correspond to the detectable and undetectable eigenvalues, respectively, of node $i$. Based on Lemma \ref{lemma:luen}, we know that node $i$ can estimate $\mathbf{z}_{\mathcal{O}_i}[k]$ asymptotically via a locally constructed Luenberger observer. It remains to show that node $i$ can recover $\mathbf{z}_{\overline{\mathcal{O}}_i}[k]$, or in other words, for each $\lambda_j \in \overline{\mathcal{O}}_i$, we need to prove that $\lim_{k\to\infty}||\hat{\mathbf{z}}^{(j)}_i[k]-\mathbf{z}^{(j)}[k]||=0$. To this end, consider a non-real eigenvalue $\lambda_j \in \Omega_{U}(\mathbf{A})$. As $\mathcal{G}$ contains a MEDAG for each $\lambda_j \in \Omega_{U}(\mathbf{A})$, the sets $\{\mathcal{L}^{(j)}_{0}, \mathcal{L}^{(j)}_{1}, \cdots, \mathcal{L}^{(j)}_{q}, \cdots \mathcal{L}^{(j)}_{T_j}\}$ form a partition of the set $\mathcal{R}$. We prove that each node in $\mathcal{R}$ can asymptotically estimate $\mathbf{z}^{(j)}[k]$ by inducting on the level number $q$.

For $q=0$, by definition of the set $\mathcal{L}^{(j)}_{0}$, all nodes in $\mathcal{L}^{(j)}_{0}$ are regular and belong to the set $\mathcal{S}_j$, i.e., $\lambda_j \in \mathcal{O}_{i}$ for each node $i$ in $\mathcal{L}^{(j)}_{0}$. Thus, by  Lemma \ref{lemma:luen}, each node in level $0$ can estimate $\mathbf{z}^{(j)}[k]$ asymptotically. Notice that for any node $i$ belonging to a level $q$, where $1 \leq  q \leq T_j$, we have $\lambda_j \in \overline{\mathcal{O}}_i$. Consider a node $i$ in $\mathcal{L}^{(j)}_{1}$ and let its error in estimation of the component $z^{(jm)}[k]$ be denoted by $e^{(jm)}_i[k] \triangleq\hat{z}^{(jm)}_i[k]-z^{(jm)}[k]$. The estimation errors of the individual components are aggregated in the vector $\mathbf{e}^{(j)}_i[k]= \mathbf{\hat{z}}^{(j)}_i[k]-\mathbf{z}^{(j)}[k]$. Subtracting $\mathbf{z}^{(j)}[k+1]$ from both sides of equation \eqref{eq:update2}, noting that $\mathbf{z}^{(j)}[k+1]=\mathbf{W}(\lambda_j)\mathbf{z}^{(j)}[k]$ (based on the dynamics given by (\ref{eqn:plant_tr})), and using \eqref{eqn:update rule1}, we obtain
\begin{equation}
\resizebox{0.8\hsize}{!}{$
\mathbf{e}^{(j)}_i[k+1]=\mathbf{W}(\lambda_j)\underbrace{\begin{bmatrix}\sum_{l \in \mathcal{M}^{(j1)}_i[k]}w^{(j1)}_{il}[k]{e}^{(j1)}_{l}[k] \\ \vdots \\ \sum_{l \in \mathcal{M}^{(j\sigma_j)}_i[k]}w^{(j\sigma_j)}_{il}[k]{e}^{(j\sigma_j)}_{l}[k]\end{bmatrix}}_{\bar{\mathbf{e}}^{(j)}_i[k]},$}
\label{eqn:errconsensus}
\end{equation}
where $\sigma_j=2a_{\mathbf{A}}(\lambda_j)$ (since $\lambda_j$ is non-real). For arriving at \eqref{eqn:errconsensus}, we used the fact that $\sum_{l \in \mathcal{M}^{(jm)}_i[k]}w^{(jm)}_{il}[k]=1$ for every component $m$ of $\mathbf{z}^{(j)}[k]$. We now analyze the error dynamics \eqref{eqn:errconsensus}. To this end, for each component $m$ of the vector $\mathbf{z}^{(j)}[k]$, we partition the set $\mathcal{N}^{(j)}_i$ into the sets $\mathcal{U}^{(jm)}_i[k]$, $\mathcal{J}^{(jm)}_i[k]$, and $\mathcal{M}^{(jm)}_i[k]$, such that the sets $\mathcal{U}^{(jm)}_i[k]$ and $\mathcal{J}^{(jm)}_i[k]$ contain $f$ nodes each, with the highest and lowest estimate values for $z^{(jm)}[k]$ respectively, transmitted to node $i$ at time-step $k$, and $\mathcal{M}^{(jm)}_i[k]$ contains the rest of the nodes in $\mathcal{N}^{(j)}_i.$ According to the update rule (\ref{eqn:update rule1}), node $i$ only uses estimates from the set $\mathcal{M}^{(jm)}_i[k]$ (which is non-empty at all time-steps based on the properties of a MEDAG) to compute the quantity $\bar{z}^{(jm)}_i[k]$. Now, for any component $m$ of $\mathbf{z}^{(j)}[k]$, consider the following two cases.
(i) $\mathcal{M}^{(jm)}_i[k] \cap \mathcal{A} = \emptyset$, \textit{i.e., there are no adversarial nodes in the set} $\mathcal{M}^{(jm)}_i[k]$: in this case, all the nodes in the set $\mathcal{M}^{(jm)}_i[k]$ are regular and belong to $\mathcal{L}^{(j)}_0$ (as $\mathcal{N}^{(j)}_i \cap \mathcal{R} \subseteq \mathcal{L}^{(j)}_{0}$). (ii) $\mathcal{M}^{(jm)}_i[k] \cap \mathcal{A}$ \textit{is non-empty, i.e., there are some adversarial nodes in the set} $\mathcal{M}^{(jm)}_i[k]$: based on the $f$-local adversarial model, it is apparent that each of the sets $\mathcal{U}^{(jm)}_i[k]$ and $\mathcal{J}^{(jm)}_i[k]$ contain at least one regular node belonging to $\mathcal{L}^{(j)}_{0}$. Let $u$ and $v$ be two such regular nodes belonging to $\mathcal{U}^{(jm)}_i[k]$ and $\mathcal{J}^{(jm)}_i[k]$, respectively. Based on the definitions of the sets $\mathcal{U}^{(jm)}_i[k]$, $\mathcal{J}^{(jm)}_i[k]$, and $\mathcal{M}^{(jm)}_i[k]$, we have $\hat{z}^{(jm)}_v[k] \leq \hat{z}^{(jm)}_l[k] \leq \hat{z}^{(jm)}_u[k]$, and hence $e^{(jm)}_v[k] \leq e^{(jm)}_l[k] \leq e^{(jm)}_u[k]$,  for every node $l \in  \mathcal{M}^{(jm)}_i[k]$. In particular, since $u,v \in \mathcal{L}^{(j)}_{0}$, it follows that for any node $l \in  \mathcal{M}^{(jm)}_i[k]$, $e^{(jm)}_{min}[k] \leq e^{(jm)}_l[k] \leq e^{(jm)}_{max}[k]$, where $e^{(jm)}_{min}[k]= \min_{u\in\mathcal{L}^{(j)}_{0}} e^{(jm)}_u[k]$ and   $e^{(jm)}_{max}[k]= \max_{u\in\mathcal{L}^{(j)}_{0}} e^{(jm)}_u[k]$. This property holds for every component $m$ of $\mathbf{z}^{(j)}[k]$. Analyzing each of the two cases, we infer that at every time-step $k$, each component of the vector $\bar{\mathbf{e}}^{(j)}_i[k]$ in \eqref{eqn:errconsensus} lies in the convex hull of the corresponding components of the error vectors $\mathbf{e}^{(j)}_{u}[k], u \in \mathcal{L}^{(j)}_0 = \mathcal{S}_j \cap \mathcal{R}$. Based on Lemma \ref{lemma:luen}, we have that $\lim_{k\to\infty} \mathbf{e}^{(j)}_{u}[k] = \mathbf{0}$, $\forall u \in \mathcal{S}_j \cap \mathcal{R}$, and hence it follows that $\hat{\mathbf{z}}^{(j)}_i[k]$ converges asymptotically to $\mathbf{z}^{(j)}[k]$ for every regular node $i$ in  $\mathcal{L}^{(j)}_{1}$. 

Suppose the result holds for all levels from $0$ to $q$ (where $ 1 \leq q \leq T_j-1 $). It is easy to see that the result holds for all regular nodes in $\mathcal{L}^{(j)}_{q+1}$ as well, by noting the following. (i) A regular node $i \in \mathcal{L}^{(j)}_{q+1}$ has $\mathcal{N}^{(j)}_i \cap \mathcal{R} \subseteq \bigcup^{q}_{r=0}\mathcal{L}^{(j)}_r$. (ii) For each $i\in\mathcal{L}^{(j)}_{q+1}$, a similar analysis reveals that at every every time-step $k$, each component of the vector $\bar{\mathbf{e}}^{(j)}_i[k]$ lies in the convex hull of the corresponding components of the error vectors $\mathbf{e}^{(j)}_u[k], u \in \bigcup^{q}_{r=0} \mathcal{L}^{(j)}_r$. The desired result then follows from the induction hypothesis. An identical argument can be sketched for a real eigenvalue $\lambda_j \in \Omega_{U}(\mathbf{A})$, and thus the result holds for any $\lambda_j \in \Omega_U(\mathbf{A})$. We arrive at the conclusion that every node $i \in \mathcal{R}$ can asymptotically estimate $\mathbf{z}^{(j)}[k]$ for every eigenvalue $\lambda_j \in \overline{\mathcal{O}}_i$. Thus, each node $i \in \mathcal{R}$ can asymptotically estimate $\mathbf{z}[k]$, and hence $\mathbf{x}[k]=\mathbf{Tz}[k]$. \end{proof}

\section{Proof of Theorem \ref{thm:tight}}
\label{app:theotight}
\begin{proof}
For sufficiency, it is easily noted that the conditions stated in the theorem guarantee termination of the MEDAG construction algorithm for every $\lambda_j\in\Omega_{U}(\mathbf{A})$. The rest of the proof for sufficiency follows identical arguments as the proof of Theorem \ref{thm:main}.

For proving necessity, we first note that the proposed algorithm summarized in Section \ref{sec:summary} is applicable only if the MEDAG construction algorithm (Algorithm 1) terminates for each $\lambda_j\in\Omega_{U}(\mathbf{A})$ and returns a subgraph $\mathcal{G}_j$ satisfying the properties of a MEDAG for all $f$-local sets containing $\mathcal{V}\setminus\mathcal{R}'$. Here $\mathcal{R}'$ denotes the set of nodes that behave regularly during the execution of Algorithm 1. Based on the hypothesis of the theorem, since $\mathcal{R}'=\mathcal{V}$, the existence of a MEDAG $\mathcal{G}_j$ $\forall \lambda_j\in\Omega_{U}(\mathbf{A})$ is necessary in this case for running the LFRE algorithm. The rest of the proof proceeds via contradiction. Suppose $\mathcal{G}$ is not strongly $(2f+1)$-robust w.r.t. $\mathcal{S}_j$ for some $\lambda_j \in \Omega_{U}(\mathbf{A})$ and yet there exists a MEDAG $\mathcal{G}_j$ for $\lambda_j$. Since $\mathcal{G}$ is not strongly $(2f+1)$-robust w.r.t. $\lambda_j$, there exists a non-empty set $\mathcal{C} \subseteq \mathcal{V}\setminus\mathcal{S}_j$ that is not $(2f+1)$-reachable. Consider the trivial $f$-local set $\mathcal{A}=\emptyset$. The subgraph $\mathcal{G}_j$ must contain a partition of $\mathcal{R}=\mathcal{V}\setminus\mathcal{A}=\mathcal{V}$ into levels that satisfy the second property of a MEDAG in Definition \ref{defn:MEDAG}. With this point in mind, let $\mathcal{C}$ be partitioned as $\mathcal{C}=\bigcup_{r=1}^{q} \mathcal{F}_r$, where $\mathcal{F}_r=\mathcal{C}\cap\mathcal{L}^{(j)}_{n_r}$ for some set of integers $\{n_1, \cdots, n_q| 1 \leq n_i \leq T_j \, \forall i\in \{1,\cdots,q\}\}$. Here, $\{\mathcal{L}^{(j)}_{n_r}\}^{q}_{r=1}$ represents a subset of the levels that partition $\mathcal{R}$ in the MEDAG $\mathcal{G}_j$ (that exists based on the hypothesis). Without loss of generality, let $n_1 < n_2 < \cdots < n_q$. Then, from the definition of a MEDAG, it follows that for any $i \in \mathcal{F}_{n_1}$, $N^{(j)}_i$ contains elements from only $\mathcal{V}\setminus\mathcal{C}$. As $\mathcal{C}$ is not $(2f+1)$-reachable, $|\mathcal{N}^{(j)}_i| < (2f+1)$, thereby violating the first property of a MEDAG in Definition \ref{defn:MEDAG}. We thus arrive at a contradiction, and the proof is complete.
\end{proof}
\section{The Strong r-robustness Property of Random Graphs}
\label{app:random}
The focus of this section is to address the following question. Given a dynamical system and an associated large-scale complex sensor network monitoring the system, under what conditions is the system and network pair $r$-feasible?  To provide an answer to this question, we study the `strong $r$-robustness' property in three relatively common random graph models for large-scale complex networks, namely the Barab\' asi-Albert (BA) preferential attachment model, the Erd\H os-R\'enyi random graph model, and the $2$-dimensional random geometric graph model.

Consider a scenario where we are given a dynamical system and an associated wireless sensor network such that the system and network pair is $r$-feasible, and hence, resilient to adversarial attacks. We wish to expand the network via addition of more sensors without disrupting the $r$-feasibility property, i.e., we intend to tolerate the same number of adversaries as earlier. As the first property in Proposition \ref{prop:properties} suggests, this can be achieved by continually adding new nodes with incoming edges from at least $r$ nodes in the existing network. The specific construction where the neighbors of a new node are selected with a probability proportional to the number of edges they already have leads to the BA preferential attachment model. Such a model is thought of as a plausible mechanism for the formation of many real-world complex networks \cite{albert}. Based on our foregoing discussion, it then follows that such real-world networks would facilitate the LFRE dynamics introduced in Section \ref{sec:LFSE}, and would hence be resilient to the worst-case attack model considered in this paper.

Next, we turn our attention to one of the most common mathematical models for large-scale networks, namely Erd\H os-R\'enyi random graphs\cite{erdos}. We denote an Erd\H os-R\'enyi random graph on $N$ nodes by $\mathcal{G}_{N,p}$, where all possible edges between pairs of different nodes are present independently and with the same probability $p$. We further note that $p$ is in general a function of the network size $N$. From the perspective of a network designer, we will be interested in answering the following questions. (i) How should the size of the source sets $\mathcal{S}_j$ (for each $\lambda_j \in \Omega_{U}(\mathbf{A})$) scale with the size of the network to maintain $r$-feasibility in an Erd\H os-R\'enyi random graph?  (ii) Which nodes should be chosen as the source nodes? Prior to answering these questions, we briefly remark on the notation to be used for the remainder of this section. The term w.h.p. (with high probability) will be used for events with probability tending to $1$ as $N \rightarrow \infty$. Given two non-negative sequences $a_{N}$ and $b_{N}$, the notation $a_{N} \ll b_{N}$ will convey the same meaning as $a_{N}=o(b_{N})$. To make use of Lemma \ref{lemma:bootstrap}, we first recall a few definitions from \cite{boot1}. Given an integer $r \geq 2$, define
\begin{equation}
T_c(N,p) \triangleq {\left(\frac{(r-1)!}{Np^r}\right)}^{\frac{1}{(r-1)}}, \, A_c(N) \triangleq \left(1-\frac{1}{r}\right)T_c(N,p).
\end{equation}

We have the following result for an Erd\H os-R\'enyi random graph model. 

\begin{proposition} 
\label{prop:random1}
Given an LTI system \eqref{eqn:plant}, a measurement model \eqref{eqn:Obsmodel}, and a network modeled by an Erd\H os-R\'enyi random graph $\mathcal{G}_{N,p}$, suppose that for each $\lambda_j \in \Omega_{U}(\mathbf{A})$, the source set $\mathcal{S}_j(N)$ is chosen randomly.\footnote{By choosing $\mathcal{S}_j(N)$ randomly, we imply that the measurement set needed to detect $\lambda_j$ is allocated to $|\mathcal{S}_j(N)|$ nodes picked uniformly at random. The notation $\mathcal{S}_j(N)$ is used to explicitly point out that the size of the source sets scales with the size of the network.} Then, the following are true.
\begin{enumerate}
\item[(i)] Let $p=p(N)$ be such that $N^{-1} \ll p \ll N^{-\frac{1}{r}}$. If for each $\lambda_j \in \Omega_{U}(\mathbf{A})$, $\frac{|\mathcal{S}_j(N)|}{A_c(N)} \geq 1+\delta$, for some $\delta >0$, and $|\mathcal{S}_j(N)| \leq \frac{N}{2}$, then $\mathcal{G}_{N,p}$ is $r$-feasible w.h.p. if and only if $Np-(ln N+(r-1)ln\, lnN) \rightarrow \infty$ as $N \rightarrow \infty.$
\item[(ii)] Let $p=p(N)$ be such that $p \gg N^{-\frac{1}{r}}$. If for each $\lambda_j \in \Omega_{U}(\mathbf{A})$, $|\mathcal{S}_j(N)| \geq r$, then $\mathcal{G}_{N,p}$ is $r$-feasible w.h.p. 
\end{enumerate}
\end{proposition}
\begin{proof} (i) If the conditions in part (i) are met, then for each $\lambda_j \in \Omega_{U}(\mathbf{A})$, $\mathcal{S}_j$ percolates via bootstrap percolation with threshold $r$ on $\mathcal{G}_{N,p}$ w.h.p. based on \cite[Theorem 3.2]{boot1}. Lemma \ref{lemma:bootstrap} then implies that $\mathcal{G}_{N,p}$ is strongly $r$-robust w.r.t. each such source set $\mathcal{S}_j$ w.h.p., i.e., $\mathcal{G}_{N,p}$ is $r$-feasible w.h.p. . The proof for part (ii) follows similarly by leveraging \cite[Theorem 5.8]{boot1} and Lemma \ref{lemma:bootstrap}. 
\end{proof}

\begin{remark}
We glean the following insights from the above result. First, we observe that if either condition (i) or condition (ii) is met, then our proposed algorithm will enable each regular node to asymptotically estimate the state of the system w.h.p. in the presence of any $\lfloor \frac{r-1}{3} \rfloor$ locally-bounded set of Byzantine adversaries. This is a direct consequence of Theorem \ref{thm:main}. The first part of Proposition \ref{prop:random1} indicates that although the source sets can be chosen randomly, their size needs to scale appropriately with the size of the network to maintain $r$-feasibility. The second part states that if the probability of edge formation is large enough, then it suffices to pick source sets of constant size equal to the bare minimum required for achieving $r$-feasibility (which equals $r$ based on part (ii) of Proposition \ref{prop:properties}).
\end{remark}

Among the three random graph models mentioned earlier, the one most relevant to our cause is the two-dimensional random geometric graph (RGG) model \cite{boot2}. RGGs are typically used to model networks where a notion of spatial proximity governs the interaction between the nodes. A wireless sensor network where randomly deployed nodes communicate with nodes only in a geographical vicinity, constitutes an ideal setup for an RGG model \cite{pottie}. We will consider a two-dimensional RGG model generated by first placing $N$ nodes randomly within the unit square $[0,1]^{2}$. Undirected edges are placed between two nodes if and only if the Euclidean distance between such nodes is at most $d(N)$, where $d(N)$ is a positive number that may depend on the network size $N$. We will denote such a RGG by $\mathcal{G}_{N,d(N)}$.

Like the Erd\H os-R\'enyi case, our focus will be on understanding how the source sets should be chosen to ensure $r$-feasibility of $\mathcal{G}_{N,d(N)}$ with high probability. To provide such a characterization, we first recall a few functions from \cite{boot2}. Let $H(x) \triangleq x \, lnx-x+1$ be defined on $[0,\infty)$ and $J(x) \triangleq lnx-1+x^{-1}$ be defined on $(0,\infty)$. Furthermore, let $J^{-1}_R: [0,\infty) \rightarrow [1,\infty)$ denote the inverse of $J(x)$ when the domain of $J(x)$ is $[1,\infty)$. We then have the following result. 

\begin{proposition} 
\label{prop:random2}
Given an LTI system \eqref{eqn:plant} and a measurement model \eqref{eqn:Obsmodel}, let the communication graph be modeled by the RGG $\mathcal{G}_{N,d(N)}$, where $d(N)=\sqrt{\frac{a\, ln N}{\pi N}}$ and $a > 1$. For each $\lambda_j \in \Omega_{U}(\mathbf{A})$, let a  node be chosen as a source node for $\lambda_j$ with a probability $p$ independently of the other nodes in the network. Let $r=\gamma a \, ln N$, where $\gamma \in (0, \frac{1}{5\pi})$. Suppose $a \geq \frac{5\pi}{H(5\pi\gamma)}$ and 
\begin{equation}
p \geq \min\left\lbrace\gamma, \frac{5\pi\gamma}{J^{-1}_{R}(\frac{1}{a\gamma})} \right \rbrace.
\end{equation}
Then, $\mathcal{G}_{N,d(N)}$ is $r$-feasible.\footnote{A RGG $\mathcal{G}_{N,d(N)}$ is connected w.h.p. for $d(N) > \sqrt{\frac{ln N}{\pi N}}$. The choice of $a > 1$ thus allows one to deal with an asymptotically connected $\mathcal{G}_{N,d(N)}$.}
\end{proposition}

\begin{proof} For each $\lambda_j \in \Omega_{U}(\mathbf{A})$, if the source set $\mathcal{S}_j$ is chosen as described above, then it percolates $\mathcal{G}_{N,d(N)}$ w.h.p. if the conditions of the proposition are met \cite[Theorem 4]{boot2}. The result then follows from Lemma \ref{lemma:bootstrap}.
\end{proof}

\begin{remark}
In an attack-prone wireless sensor network, one might be interested in tolerating $f$-local adversarial sets where the paramater $f$ scales with the size of the network. Such a possibility is captured by  Proposition \ref{prop:random2}, based on which, $\lfloor \frac{\gamma a \, ln N-1}{3} \rfloor$-local Byzantine adversarial sets can be accounted for by our proposed algorithm.
\end{remark}
\end{document}